\documentclass[acmsmall,nonacm]{acmart}

\acmJournal{PACMPL}
\acmVolume{1}
\acmNumber{CONF} 
\acmArticle{1}
\acmYear{2018}
\acmMonth{1}
\acmDOI{} 
\startPage{1}

\setcopyright{none}

\usepackage{mdframed}
\usepackage{multicol}
\usepackage{nicefrac}
\usepackage{xfrac}
\usepackage{listings}
\usepackage{commands}
\usepackage{bcprules}
\usepackage{fdsymbol}

\usepackage{thmtools}
\usepackage{thm-restate}

\usepackage{amsmath}

\usepackage[frozencache,cachedir=.]{minted}
\setminted{fontsize=\small,breaklines=true}
\newcommand{\todo}[1]{{\color{red}(\textbf{todo: }{#1})}}

\usepackage{mathpartir}

\usepackage{booktabs}   
\usepackage{subcaption} 

\begin{document}

\title{Formalizing Box Inference for Capture Calculus}


\author{Yichen Xu}
\affiliation{
  \institution{EPFL}            
  \country{Switzerland}                    
}
\email{yichen.xu@epfl.ch}          

\author{Martin Odersky}
\affiliation{
  \institution{EPFL}           
  \country{Switzerland}                   
}
\email{martin.odersky@epfl.ch}         

\begin{abstract}
  Capture calculus has recently been proposed as a solution to effect checking,
  achieved by tracking the captured references of terms in the types.
  Boxes, along with the \emph{box} and \emph{unbox} operations,
  are a crucial construct in capture calculus,
  which maintains the hygiene of types and improves the expressiveness of polymorphism over capturing types.
  Despite their usefulness in the formalism,
  boxes would soon become a heavy notational burden for users when the program grows.
  It thus necessitates the inference of boxes
  when integrating capture checking into a mainstream programming language.
  In this paper, we develop a formalisation of box inference for capture calculus.
  We begin by introducing a semi-algorithmic variant of the capture calculus,
  from which we derive an inference system where typed transformations are applied
  to complete missing box operations in programs.
  Then, we propose a type-level system that performs provably equivalent inference on the type level,
  without actually transforming the program.
  In the metatheory, we establish the relationships between these systems and capture calculus,
  thereby proving both the soundness and the completeness of box inference.
\end{abstract}



\maketitle



%
%

\section{Introduction}
\label{sec:introduction}



Capture calculus (\system{}) \cite{Odersky2022ScopedCF} is recently proposed as
a solution to \emph{effect checking} \cite{effekt1,effekt2}.
Effect checking goes beyond tracking the \emph{shapes} of values (e.g. whether a value is an integer, a list, or a function),
aiming at statically predicting the \emph{effects} of a program.
Throwing an exception, performing I/O operations, or waiting for events, are examples of effects.
Capture calculus takes the object-capability-based approach to modelling effects,
wherein effects are performed via object capabilities \cite{capability1,capability2,effekt1,effekt2}.
Object capabilities are regular program variables in the calculus,
and effect checking is achieved by tracking \emph{captured} variables in types.
\begin{minted}{scala}
  val file: {*} File = open("test.txt")
  def prog(): Unit = {
    console.write("Hello!")
    file.write("Hello, again!")
  }
\end{minted}
For instance, in this example,
the \mintinline{scala}{prog} function is typed as \mintinline{scala}{{console, file} () -> Unit}.
This is a \emph{capturing type},
which has a \emph{capture set} (\mintinline{scala}{{console, file}})
in addition to the shape type (\mintinline{scala}{() -> Unit}).
The shape type describes the shape of the function,
whereas the capture set indicates that the function \emph{captures} the variables \texttt{console} and \texttt{file}.
By viewing both \texttt{file} and \texttt{console} as {capabilities},
the capture set \mintinline{scala}{{console, file}} essentially characterizes effects of this function,
i.e.
accessing the \texttt{file}, and performing \texttt{console} I/O.
A type without or with an empty capture set is considered to be \emph{pure}.
In \system{}, every capability obtains its authority
from some other capabilities that it captures,
and \texttt{*} is the \emph{root capability},
from which all other capabilities are ultimated derived.

\system{} has \emph{boxes} in both terms and types,
which serve as a crucial formal device.
We explain the role of boxes by
ecucidating the \emph{capture tunneling} mechanism in \system{},
which heavily involve boxes.
Capture tunneling facilitates the interaction between capture checking and type polymorphism.
We start by considering the following example \cite{Odersky2022ScopedCF},
which defines a pair class.
\begin{minted}{scala}
  class Pair[+A, +B](x: A, y: B):
    def fst: A = x
    def snd: B = y
\end{minted}
Problems arise when we attempt to construct a pair with impure values:
\begin{minted}{scala}
  def x: {ct} Int -> String
  def y: {fs} Logger
  def p = Pair(x, y)
\end{minted}
The problem is: what should the type of \mintinline{scala}{p} capture?
\mintinline{scala}{{ct, fs}} is a sound option,
but as the program grows,
the capture sets will quickly accumulate
and become increasingly verbose and imprecise.
Moreover, it undermines the expressiveness of type polymorphism \cite{Odersky2022ScopedCF},
which is illustrated in the following example:
\begin{minted}{scala}
  def mapFirst[A, B, C](p: Pair[A, B], f: A => C): Pair[C, B] =
    Pair(f(p.fst), p.snd)
\end{minted}
Since pairs can be impure, we have to annotate the argument \texttt{p} to capture \mintinline{scala}{{*}},
as well as the return type,
This results in imprecise types.
To address these issues,
capture tunneling proposes to prevent the propagation of captured capabilities at the instantiation of type variables,
and only pop out these capabilities when the value is used.
Specifically in this example, the type of \texttt{p} will have an empty capture set,
while the following function:
\begin{minted}{scala}
  () => p.fst     : {ct} () -> {ct} Int -> String
\end{minted}
which \emph{accesses} the member,
will capture \mintinline{scala}{{ct}}.

\system{} achieves capture tunneling
by retricting type variable instances to shape types
and utilizing boxes.
The box hides the capture set and stops the propagation of the captures,
and the unbox reveals the boxed capture set.
Specifically,
a capturing type \mintinline{scala}{{x1, ..., xn} S}
can be encapsulated in a boxed type $\Box$ \mintinline{scala}{{x1, ..., xn} S},
which conceals the capture set and turns the capturing type into a shape/pure type.
A boxed term introduces a box in the type.
which can then be removed by the unbox expression when the term under the box is to be accessed.
In our example, the expression for creating \mintinline{scala}{p} will then look like:
\begin{minted}[escapeinside=||]{scala}
  def p = Pair[|$\new\Box$| {ct} Int -> String, |$\new\Box$| {fs} Logger](|$\new{\Box}$ x|, |$\new{\Box}$| y)
\end{minted}
Accessing the term under the box requires the use of an \emph{unbox} expression which unveils the hidden capture set,
like the following:
\begin{minted}[escapeinside=||]{scala}
  () => {ct} |$\multimapinv$| p.fst     : {ct} () -> {ct} Int -> String
\end{minted}
Here, \mintinline[escapeinside=||]{scala}{{ct} |$\multimapinv$| p.fst} is the syntax for unboxing.

Moreover, knowing that type variable instances are pure,
the \texttt{mapFirst} function can be written as it is in the previous example.
It is concise and at the same time expressive enough to work polymorphically on pairs.

Boxes are an essential formal device within the \system{} framework.
Yet,
being a theoretical tool which exist exclusively on the conceptual level,
boxes do not exist at runtime
and have no effects in program semantics.
When putting capture checking into practice,
mandating the annotation of the program with appropriate boxes
is unnatural,
posing an unnecessarily heavy notational burden.
Hence,
the introduction of box inference becomes a pressing need.
Scala, which has a practical implementation of capture checking,
does not have boxes in the surficial language.
Instead,
the capture checker infer the box operations
wherever they are necessary to make the input program well-typed.
Such an approach enables complete concealment of boxes as a technical detail of the underlying theory
and liberates users from the knowledge of their existence.
This aspect is crucial to integrating capture calculus into a mainstream language.


Despite that the capture checker in Scala has already implemented box inference,
we still lack the theoretical foundation of box inference.
In fact, box inference is unexpectedly non-trivial to be implemented correctly.
To demonstrate this point,
let us inspect how box inference should work in the following example,
where we want to execute a list of IO operations:
\begin{minted}{scala}
  class List[+A]:
    ...
    def map[B](op: A => B): List[B] = ...
	
  def ops: List[{io} () -> Unit] = ...
  def runOp(op: {io} () -> Unit): Unit = op()

  def runOps: {io} () -> Unit = () => ops.map[Unit](runOp)
\end{minted}
Here, we have \texttt{ops} which is a list of IO operations,
an function \texttt{runOp} which takes an IO operation and execute it,
and the function \texttt{runOps} which execute all the IO operations in the list.
In \system{}, \mintinline{scala}{A => B} is an alias of \mintinline{scala}{{*} A -> B},
which essentially means an impure function
that could possibly perform all kinds of effects.
This behavior is achieved by the subcapturing mechanism of \system{},
which we are going to explain later in the backgrounds.
We expect to infer a box in the type of \texttt{ops},
making it \mintinline[escapeinside=||]{scala}{List[|$\Box$| {io} () -> Unit]}.
Otherwise, the program will be syntactically incorrect,
as the type variable should be instantiated to a shape type.
The other place involving box inference,
which is also where the nontriviality arises,
is \mintinline{scala}{ops.map(runOp)}.
In this application,
\texttt{map} is expecting an argument of type
\mintinline[escapeinside=||]{scala}{(|$\Box$| {io} () -> Unit) => Unit},
whereas the type of \texttt{runOp} is
\mintinline{scala}{({io} () -> Unit) => Unit}.
To correctly infer the boxes,
we suprisingly have to perform eta-expansions.
The box inference should transform \texttt{runOp}
into \mintinline[escapeinside=||]{scala}{(op: |$\Box$| {io} () -> Unit) => runOp({io} |$\multimapinv$| runOp)}.
It is also noteworthy that
although the body of \texttt{runOps} in the surfacial does not capture \texttt{io} in any way,
as both \texttt{ops} and \texttt{runOp} has an empty capture set,
\texttt{runOps} in the elaborated program does capture the \texttt{io} capability
due to the insertion of the unbox operation during box inference.
The necessity of eta-expansion
and its effect on the capture sets
both illustrate the non-triviality of box inference.

To this end, in this paper we propose a formal foundation of box inference,
which can lead to a deeper and principled understanding of it,
and facilitate the development of a reliable and sound box inference algorithm in the compiler.
Actually, our formalization of box inference eventually fostered two pull-requests in the Scala 3 compiler
that make fixes and improvements to the box inference algorithm.
Our contributions can be summarized as follows:
\begin{itemize}
  \item We develop \ccAlgo{}, A syntax-directed variant of \cc{}, which straightforwardly gives rise to a semi-algorithm for typechecking.

  \item Based on the semi-algorithmic variant, we propose a box inference system \ccAdp{} that completes the missing box operations in the program while typechecking it.

  \item Since box operations has zero runtime effects, transforming the program during box inference is a waste.
    We take one step further and propose a type-level box inference system \ccAdpt{} that
    performs provably the same reasoning but only operating on types
    without referring to terms.

  \item We develop a metatheory for these systems, proving the equivalence between them and capture calculus.

  \item
    To provides insights into the implementation of  capture checking and box inference in other languages,
    we discuss the implementation of box inference in Scala 3,
    which has capture checking available as an experimental language feature.
\end{itemize}

\section{Background}

In the background, we will briefly introduce the key ideas of capture calculus,
which is the base framework of our work.
Then, we will discuss the problem of box inference in more details.

\subsection{Introduction to Capture Calculus}

In Section \ref{sec:introduction}
we have made a brief overview of capture calculus,
Now we will delve deeper into the concepts of this system,
providing additional explanations to help user gain a better understanding
and touching on other important aspects.

\paragraph{Capture sets and capturing types.}
One fundamental concept of \system{} is to track captured variables in types.
\system{} stratifies types into
\begin{itemize}
\item \emph{Shape types} which describes the shape of values,
  e.g. functions, type variables, and boxed types.

\item \emph{Capturing types}
  which, in addition to the shape type, have a \emph{capture set}
  specifying the captured variables of the associated term.
\end{itemize}
A capture set is a finite set of variables,
e.g. \mintinline{scala}{{file, console}} in previous examples.
Apart from the program variables in scope,
it can include the special variable \texttt{*},
which is the root capability.
The capture set in the type predicts the captured variables of a value.
In the previous example,
the function that accesses \texttt{file} and \texttt{console}
will have both variables in the capture set of its type.
Note that only the reference to capabilities in the environment
will be considered as captures, similar to the idea of free variables.
The following example from \citeN{Odersky2022ScopedCF} illustrates the idea:
\begin{minted}{scala}
def test(fs: FileSystem): {fs} String -> Unit =
  (x: String) => Logger(fs).log(x)
\end{minted}
\mintinline{scala}{test} will be typed as
\mintinline{scala}{{} (fs: FileSystem) -> {fs} String -> Unit}.
The capability \texttt{fs} is only captured by the returned closure.

\paragraph{Subcapturing}
\system{} employs subcapturing to compare capture sets.
Given capture sets $C_1$ and $C_2$,
we say $C_1$ is a subcapture of $C_2$ (written $C_1 <: C_2$)
if every element $x$ in $C_1$ has been accounted by $C_2$.
$x$ is considered to be accounted by $C_2$
either if $x$ is an element in $C_2$,
or the capture set $C$ of $x$ is a subcapture of $C_2$ ($C <: C_2$).
Let us inspect a concrete example.
In the following environment:
\begin{minted}{scala}
file    : {*} File
console : {*} Console
op      : {file, console} () -> Unit
l       : {console} Logger
\end{minted}
the following subcapture relations hold:
\begin{minted}{scala}
  {l}    <: {file}          <: {*}
  {op}   <: {file, console} <: {*}
  {file} <: {file, console} <: {*}
  {l}    <: {file, console} <: {*}
\end{minted}
Since every capability is ultimately derived from \texttt{*},
it is a supercapture of all capture sets.
As mentioned in Section \ref{sec:introduction},
a function that capture \texttt{*}
(e.g. \mintinline{scala}{{*} () -> Unit} or \mintinline{scala}{() => Unit})
can perform any sort of effects.

\paragraph{Escape checking}
Besides capture tunneling introduced in Section \ref{sec:introduction},
the other crucial mechanism in \system{} is escape checking \cite{Odersky2022ScopedCF}.
In the following code, we create a file handle and pass it to a function, and close it at the end.
\begin{minted}{scala}
  def usingLogFile[T](op: (file: {*} File) => T): T = {
    val file = new File("output.log")
    try op(file)
    finally
    file.close()
  }

  def good = () => usingLogFile(f => f.write(0))

  def bad = () => {
    val later = usingLogFile(f => () => f.write(0))
    later()
  }
\end{minted}
Here, \mintinline{scala}{good} should be allowed
whereas \mintinline{scala}{bad} should be rejected.
\mintinline{scala}{bad} returns the scoped capability out of its valid scope
and try to use the escaped capability.
We inspect the typing of the \mintinline{scala}{usingLogFile} application step-by-step
to illustrate how it is rejected.
\begin{itemize}
\item The type of \mintinline{scala}{f => () => f.write(0)}
  is \mintinline{scala}{(f : {*} File) -> {f} () -> Unit}.
\item When instantiating the type argument of \mintinline{scala}{usingLogFile},
  we can only instantiate it to \mintinline[escapeinside=||]{scala}{|$\Box$| {*} () -> Unit}.
\item The type of \mintinline{scala}{later} is \mintinline[escapeinside=||]{scala}{|$\Box$| {*} () -> Unit}.
\item Unboxing a term whose capture set under the box contains \texttt{*} is disallowed.
\item The escaped capability cannot be used, and therefore \mintinline{scala}{later()} will be rejected.
\end{itemize}

\subsection{The Box Inference Problem}

\newcommand{\redtexttt}[1]{{\color{red}\texttt{#1}}}
\newcommand{\bluetexttt}[1]{{\color{blue}\texttt{#1}}}
\newcommand{\redbox}{{\color{red}{\ensuremath{\Box}}}}
\newcommand{\bluebox}{{\color{blue}{\ensuremath{\Box}}}}

Now we take a close look at the box inference problem.
Given a program with possibly missing box constructs,
box inference aims to find a way to complete the program with box operations,
such that the completed one is well-typed in the base calculus.
For instance, for the aforementioned \mintinline{scala}{Pair} example,
the user is allowed to write the program \mintinline{scala}{new Pair[{fs} String, {console} Int](x, y)}.
The missed boxes should be inferenced by the compiler, producing the well-typed expression
\mintinline[escapeinside=||]{scala}{new Pair[|\redbox{}| {fs} String, |\redbox{}| {console} Int](|\bluebox{}| x, |\bluebox{}| y)}.
There are two kinds of boxes getting inferenced: the \redtexttt{box}es in types, and \bluetexttt{box}es in terms.

Boxes in types can be trivially inferenced based on the syntactic class of a type.
In other words, when we are instantiating a type variable with a capturing type,
it should get boxed.
In this example, we find out that two impure types are supplied in the type application,
and infer the two \redtexttt{box}es for them.
By contrast, inferencing the boxes in terms is non-trivial.
Box inference over terms involves eta-expansion and induces complex effects on the capture sets.
In this paper, we focus on \emph{box inference over terms}.
and presents three systems
that in together solve this problem.
Specifically, given a program $t$, our systems attempts to complete the missing box constructs \emph{in terms} so that the completed program $t^\prime$ becomes well-typed as $T$.

\section{Key Ideas of Box Inference}

In this section, we discuss the box inference systems informally
and present the key ideas of them.

\subsection{Semi-Algorithmic Typing}

Our systems aim to formalize the box inference procedure.
Therefore, we expect the systems to be semi-algorithmic (or syntax-directed), so that they straightforwardly give rise to procedures for box inference.
As the original capture calculus is declarative, the first step of our work is to derive a syntax-directed variant of the calculus.

The three major hindrance of syntax-directness is
the transitivity rule in subtyping,
the subsumption rule in typing,
and the let rule where the choice of the type that avoids the local variable is ambiguous.
The idea of the syntax-directed variant is to inline the of first two rules in subtyping and typing respectively,
and derive a avoidance procedure for the let rule.
The resulted system is equivalent to the declarative one, as proven in the metatheory.
The semi-algorithmic calculus serves as the basis of the two box inference systems we are going to develop.

\subsection{Box Adaptation}

Before explaining \emph{how} we do box inference, let's first talk about \emph{when} to do it.
In general, box infernece's input include (1) a term $t$ of type $T$ to be adapted,
and (2) the expected type $U$.
We observe that this is exactly the same situation as subtyping checks.
In the semi-algorithmic variant of \system{}, the subtyping check is only when typechecking applications.

Specifically, given the term \mintinline{scala}{f(x)}, if \texttt{f} is a function that receives an arguments of type \texttt{U}, and \texttt{x} is typed as \texttt{T}, we invoke subtype check to answer the question
\emph{whether $T$ is a subtype of $U$},
or \emph{whether a value of type $T$ can be used as $U$}.
On the other hand, in box inference the question is
\emph{whether a value of type $T$ can be adapted to $U$ by completing missing box operations}.
To this end, we employ \textbf{box adaptation} as a plug-in replacement for subtype checks, which inserts box operations in addition to performing regular subtyping.
It answers exactly the above question,
and returns the adapted term if answer is yes.
Box adaptation is where box inference actually happens in our system.

The basic idea of box adaptation is to compare the types recursively as subtyping does, but transform the term when it sees a mismatch in boxes.
For example, when adapting the term \texttt{x} of type \mintinline{scala}{{fs} String} against the expected type \mintinline[escapeinside=||]{scala}{|$\Box$| {fs} String}, box adaptation discovers a mismatch in boxes, so it inserts a box, and returns the adapted term \mintinline[escapeinside=||]{scala}{|$\Box$| x}.
The design of the box adaptation judgement will be explained in details in Section \ref{sec:box-adaptation}.

As the term is transformed during box adaptation, the typing judgment of the box inference system not only derives a type $T$ of a input term $t$, but also returns the adapted term $t^\prime$ which is well-typed as $T$ in \system{}.



\subsection{Adapting Functions with Eta-Expansion}

Box-adapting functions requires eta-expansion.
In Section \ref{sec:introduction} we have already seen an example:
given \texttt{f} of type \mintinline{scala}{({io} () -> Unit) -> Unit},
and the expected \mintinline[escapeinside=||]{scala}{(|$\Box$| {io} () -> Unit) -> Unit},
we have to eta-expand the function \texttt{f} so that we can box-adapt its argument.
After eta-expansion \texttt{f} into \texttt{x => f(x)},
we find that the type of the argument \texttt{x} is \mintinline[escapeinside=||]{scala}{|$\Box$| {io} () -> Unit},
where as to apply the function the expected type is \mintinline{scala}{{io} () -> Unit}.
We therefore insert an unbox and returns the adapted term \mintinline[escapeinside=||]{scala}{x => f({io} |$\multimapinv$| x)}.

Eta-expansion and box adaptation not only changes the boxes in types, but also changes the capture sets.
In the aforementioned example, the type of the adapted form becomes \mintinline{scala}{{io} ({io} Unit -> Unit) -> Unit}.
Note that box adaptation charges the outermost capture sets with additional captured reference \texttt{io}.


\subsection{Type-Level Box Inference}

The transformation performed by box inference (including the box/unbox operations and the eta-expansion) does not have any runtime effects.
The only point that matters is whether or not the value of one type could be adapted to another type based on box inference.
Therefore, actually computing the transformed term could be a waste.
It is natural to ask: can we check the possibility of such adaptation without actually keeping track of the result term?
The answer is yes.

The type-level box inference system does this by predicting the type-level effects of box adaptation, without actually transforming the terms.
For example, when we infer a box for the variable \mintinline{scala}{x} of type \mintinline{scala}{{fs} String}, the type-level system predicts that the type of the adapted term becomes \mintinline[escapeinside=||]{scala}{|$\Box$| {fs} String}, without computing the term \mintinline[escapeinside=||]{scala}{|$\Box$| x}.

To properly predict the type-level effects of box inference,
the system has to track not only the change of boxes in types, but also the capture sets.
As we have seen before, eta-expansion will charge the capture sets with additional references.
Considering the example of executing wrapped executions.
\begin{minted}[escapeinside=||]{scala}
  type Wrapper[T] = [X] -> (f: {*} T -> X) -> X

  def run(wrapped: Wrapper[|$\Box$| {io} Unit -> Unit]) =
    def f(op: {io} Unit -> Unit): Unit = op(unit)
    wrapped(f)
\end{minted}
When typing the application \mintinline{scala}{wrapped(f)},
we adapt \texttt{f} of type \mintinline{scala}{(op: {io} Unit -> Unit) -> Unit} to the expected type
\mintinline[escapeinside=||]{scala}{{*} (op: |\tbox| {io} Unit -> Unit) -> Unit}.
The term-level system adapts the term into \mintinline[escapeinside=||]{scala}{x => f({io} |\tunbox| x)}, as discussed before,
whose type is \mintinline[escapeinside=||]{scala}{{io} (op: |\tbox{}| {io} Unit -> Unit) -> Unit}.
If we define \mintinline{scala}{Wrapper[T]} as \mintinline{scala}{type Wrapper[T] = [X] -> (f: {} T -> X) -> X},
which requires the function \texttt{f} to be pure, the program should be rejected,
since after box adaptation the function \texttt{f} is impure (and it should be, since it runs a wrapped IO operation).
The type-level box inference has to keep track of the fact that the capture set of the function type is charged with \texttt{io}, in addition to predicting the box in the result type.
Furthermore, the function \texttt{run} captures \texttt{io} as well after box inference due to the box adaptation that happened in its closure.
The type-level system has to take care of these effects on captured variables when typing the closures.

\section{Formal Presentation of Box Inference}

\subsection{Semi-Algorithmic Capture Calculus \ccAlgo{}}

We first present the semi-algorithmic variant \ccAlgo{},
based on which our two box inference systems are developped.
\ccAlgo{} is derived from \system{} by making the rules syntax-directed.
\system{} is extended from System F$_{<:}$
with the added captrue checking constructs including capture sets, capturing types and boxes \cite{Odersky2022ScopedCF}.
It is in monadic normal form and dependently typed, allowing for capture sets in types to reference program variables.

\textbf{Syntax.}
\systemAlgo{} shares the same syntax with \system{}.
The syntax is presented in Figure \ref{fig:syntax}.
The main differences between the syntax of \system{} and System F$_{<:}$ are highlighted in \tnew{grey boxes}.
There are the box $\tBox{x}$, unbox $\tUnbox{C}{x}$ expressions,
and boxed types $\tBox{T}$.
Also, the types are now annotated with a capture set $C$, which is a set of variables.
Shape types are regular System F$_{<:}$ types without capturing information.
Shape types can be used as types directly by assuming the equivalence $S \equiv \tCap{\set{}}{S}$.
The readers may find more details about the syntax of capture calculus in \cite{Odersky2022ScopedCF}.

\textbf{Captured variables.}
The function $\cv{\cdot}$ (in the \algorn{abs}, and the \algorn{tabs} rules)
computes the set of captured variables of a term \cite{Odersky2022ScopedCF}.
$\cv{t}$ is mostly identical to the free variables of $t$,
except that the references inside boxes are ignored.
The definition of $\cv{\cdot}$ \cite{Odersky2022ScopedCF} is given in Definition \ref{def:captured-variables}.
$\tBox{x}$ and $\tUnbox{C}{x}$ hides and reveals the captured variables respectively.
$\cv{v}$ is counted in $\cv{\tLet{x}{v}{t}}$ only when $t$ captures $x$.
\begin{definition}[Captured variables]  \label{def:captured-variables}
  $\cv{t}$ computes the set of captured variables of a term $t$,
  and is defined as follows:
  \begin{alignat*}{4}
      &\cv{\tLambda{x}{T}{t}} &\quad=\quad &\cv{t} / \set{x} \\
      &\cv{\tTLambda{X}{S}{t}} &\quad=\quad &\cv{t} \\
      &\cv{x} &\quad=\quad &\set{x} \\
      &\cv{\tLet{x}{v}{t}} &\quad=\quad &\set{\cv{t}}  \quad \text{if $x \notin \cv{t}$} \\
      &\cv{\tLet{x}{t}{u}} &\quad=\quad &\cv{t} \cup \cv{u} / \set{x} \\
      &\cv{x\,y} &\quad=\quad &\set{x, y} \\
      &\cv{x[S]} &\quad=\quad &\set{x} \\
      &\cv{\tBox{x}} &\quad=\quad &\set{} \\
      &\cv{\tUnbox{C}{x}} &\quad=\quad &\set{x} \cup C \\
  \end{alignat*}
\end{definition}

\newcommand{\bbox}[1]{\mbox{\textbf{\textsf{#1}}}}
\newcommand{\OR}{\mathop{\ \ \ |\ \ \ }}

\begin{figure}[t]
  \vspace{0.3em}
  \begin{center}\rule{1.0\textwidth}{0.4pt}\end{center}
  \vspace{0.3em}

  $\begin{array}[t]{llll@{\hspace{8mm}}l}
    \bbox{Variable} & \multicolumn{2}{l}{x, y, z} \\
    \bbox{Type Variable} & \multicolumn{2}{l}{X, Y, Z} \\[1em]
    \bbox{Value} & v, w & ::= & \lambda(x : T)t
        \OR \lambda[X <: S]t
        \OR \new{\tBox{x}} \\

    \bbox{Answer} & a & ::= & v \OR x \\

    \bbox{Term} &s, t& ::= & a
          \OR  x\,y
          \OR  x\,[S]
          \OR  \tLet x s t
          \OR  \new{\tUnbox{C}{x}}\\

    \bbox{Shape Type} & S & ::= & X
          \OR  \top
          \OR  \forall(x : U) T
          \OR  \forall[X <: S] T
          \OR  \new{\tBox{T}} \\

    \bbox{Type} & T, U & ::= & S \OR \new{\tCap{C}{S}} \\

    \bbox{Capture Set} & \new{C} & ::= & \new{\{x_1, \ldots, x_n\}}
    \end{array}$

\caption{\label{fig:syntax} Syntax of \system{} and the box inference systems \cite{Odersky2022ScopedCF}}
\vspace{0.3em}
\begin{center}\rule{1.0\textwidth}{0.4pt}\end{center}
\vspace{0.3em}
\end{figure}

\begin{wide-rules}

\textbf{Algorithmic typing \quad $\typAlgDft{t}{T}$}

\begin{multicols}{2}

\infrule[\ruledef{alg-var}]
  {\G(x) = \tCap{C}{S}}
  {\typAlgDft{x}{\tCapSet{x}{S}}}

\infrule[\ruledef{alg-abs}]
{\typAlg{\extendG{x}{U}}{t}{T} \\ \wfTyp{\G}{U}}
{\typAlgDft{\tLambda{x}{U}{t}}{\tCap{\exclude{\cv{t}}{x}}{\tForall{x}{U}{T}}}}

\infrule[\ruledef{alg-tabs}]
{\typAlg{\extendGT{X}{S}}{t}{T} \\ \wfTyp{\G}{S}}
{\typAlgDft{\tTLambda{X}{S}{t}}{\tCap{\cv{t}}{\tTForall{X}{S}{T}}}}

\infrule[\ruledef{alg-app}]
{\typUpDft{x}{\tCap{C}{\tForall{z}{T}{U}}} \\
  \typAlgDft{y}{T^\prime}\\
  \new{\subAlgoDft{T^\prime}{T}}}
{\typAlgDft{x\ y}{\tSubst{z}{y}{U}}}

\infrule[\ruledef{alg-tapp}]
{\typUpDft{x}{\tCap{C}{\tTForall{X}{S}{T}}} \\
  \new{\subAlgoDft{S^\prime}{S}}}
{\typAlgDft{x\ S^\prime}{\tSubst{X}{S^\prime}{U}}}

\infrule[\ruledef{alg-box}]
{\typAlgDft{x}{\tCap C S} \\ C \subseteq \dom{\G}}
{\typAlgDft{\tBox{x}}{\tBox{\tCap C S}}}

\infrule[\ruledef{alg-unbox}]
{\typUpDft{x}{C_x\,\tBox{\tCap C S}} \\ C^\prime \subseteq \dom{\G}\\
  \subAlgoDft{C}{C^\prime}}
{\typAlgDft{\tUnbox{C^\prime}{x}}{\tCap{C}{S}}}

\infrule[\ruledef{alg-let}]
{\typAlgDft{s}{T} \\ \typAlg{\extendG{x}{T}}{t}{U}\\
  \new{U^\prime = \avoidOp{x}{\cv{T}}{U}}}
{\typAlgDft{\tLet{x}{s}{t}}{U^\prime}}

\end{multicols}


\textbf{Algorithmic subtyping \quad $\subAlgoDft{T}{U}$}

\begin{multicols}{2}

\infax[\ruledef{alg-refl}]
{\subAlgoDft{\new{X}}{\new{X}}}

\infrule[\ruledef{alg-tvar}]
{X <: S \in \G \\ \new{\subAlgoDft{S}{S^\prime}}}
{\subAlgoDft{X}{S^\prime}}

\infax[\ruledef{alg-top}]
{\subAlgoDft{S}{\top}}

\infrule[\ruledef{alg-fun}]
{\subAlgoDft{U_2}{U_1} \\
  \subAlgo{\extendG{x}{U_2}}{T_1}{T_2}}
{\subAlgoDft{\tForall{x}{U_1}{T_1}}{\tForall{x}{U_2}{T_2}}}

\infrule[\ruledef{alg-tfun}]
{\subAlgoDft{S_2}{S_1} \\
  \subAlgo{\extendGT{X}{S_2}}{T_1}{T_2}}
{\subAlgoDft{\tTForall{X}{S_1}{T_1}}{\tTForall{X}{S_2}{T_2}}}

\infrule[\ruledef{alg-capt}]
{\subAlgoDft{C_1}{C_2} \\
  \subAlgoDft{S_1}{S_2}}
{\subAlgoDft{\tCap{C_1}{S_1}}{\tCap{C_2}{S_2}}}

\infrule[\ruledef{alg-boxed}]
{\subAlgoDft{T_1}{T_2}}
{\subAlgoDft{\tBox{T_1}}{\tBox{T_2}}}

\end{multicols}

\textbf{Subcapturing \quad $\subDft{C_1}{C_2}$}

\begin{multicols}{3}

\infrule[\ruledef{sc-var}]
{\G(x) = \tCap{C}{S}\\ \subDft{C}{C_2}}
{\subDft{\set{x}}{C_2}}

\infrule[\ruledef{sc-set}]
{\subDft{\set{x_1}}{C_2}\\
\cdots\\
\subDft{\set{x_n}}{C_2}}
{\subDft{\set{x_1,\cdots,x_n}}{C_2}}

\infrule[\ruledef{sc-elem}]
{x \in C}
{\subDft{\set{x}}{C}}

\end{multicols}

\caption{Algorithmic Type System (\ccAlgo{})}
  \label{fig:algo-typing}

\end{wide-rules}

\begin{wide-rules}






\textbf{Type variable widening \quad $\typUpDft{X}{S}$}

\begin{multicols}{2}

\infrule[\ruledef{widen-shape}]
{X <: S \in \G \\ S \notin \mathcal{X}}
{\typUpDft{X}{S}}

\infrule[\ruledef{widen-tvar}]
{X <: Y \in \G \\ \typUpDft{Y}{S}}
{\typUpDft{X}{S}}

\end{multicols}

\textbf{Widen Variable Typing \quad $\typUpDft{x}{\tCap C S}$}

\begin{multicols}{2}

\infrule[\ruledef{var-widen}]
{\typAlgDft{x}{\tCap{C}{X}}\\
  \typUpDft{X}{S}}
{\typUpDft{x}{\tCap{C}{S}}}

\infrule[\ruledef{var-lookup}]
{\typAlgDft{x}{\tCap{C}{S}} \\
  S \notin \mathcal{X}}
{\typUpDft{x}{\tCap C S}}

\end{multicols}

\caption{Variable Typing of System \ccAlgo{}}
  \label{fig:var-typing}

\end{wide-rules}

\textbf{Syntax-directed typing and subtyping.}
The type system of \ccAlgo{} is shown in Figure \ref{fig:algo-typing},
where the grey boxes highlight the main differences from \system{}.
The rules are made syntax-directed by inlining \emph{subsumption} in typing, and \emph{transitivity} in subtyping.
Also, we make use of the \textsf{avoid} function in \ruleref{alg-let}, that computes the least supertype of a type $U$ that does not mention the local variable $x$.
We will discuss avoidance in more details later.
The \emph{subsumption} rule in typing is inlined to two application rules \algorn{app} and \algorn{tapp}.
They check whether the argument type conforms to the parameter type, or the type variable bounds.
In subtyping, the transitivity is inlined in the \algorn{tvar} rule.
Also, the \algorn{refl} rule only applies on type variables.
We prove the equivalence between \systemAlgo{} and \system{} in metatheory.
We show the subcapturing rules as well for the completeness of the presentation,
though they are unchanged from \system{}.

\textbf{Avoidance.}
In \algorn{let} we invoke the \textsf{avoid} function to
algorithmically construct the least super type of the body type such that the locally-bound variable is not mentioned.
The \textsc{let} rule of the original \system{} is shown below,
which requires that the result type avoids $x$,
while the choice of $U$ is ambiguous and thus undermines the semi-algorithmics.
\infrule[let]
{\typDft{s}{T} \\ \typ{\extendG{x}{T}}{t}{U}\\
  \new{x \notin \fv{U}}}
{\typDft{\tLet{x}{s}{t}}{U}}
$\avoidOp{x}{\cv{T}}{U}$ computes the aforementioned smallest $x$-avoiding supertype $U^\prime$ for $U$.
The avoidance function is a traversal of the input type,
where we approximate $x$ to $C_x$ at covariant occurrences and to $\set{}$ at contravariant occurrences.



\textbf{Type variable widening.}
In \algorn{app}, \algorn{tapp} and \algorn{unbox}
the $\typUpDft{x}{T}$ judgement looks up variable $x$ in the environment,
and widen its type to a concrete type if it is bound to a type variable.
For instance, given an application $x\,y$
where $x$ is typed as a type variable $X$
whose bound is a function.
We expect the type of $x$ to be a function,
but in lack of a separate subsumption rule,
we can not widen the type of $x$ to the function via subtyping.
To solve this problem, we employ the $\typUpDft{x}{T}$ judgment here
to widen the type variable to a concrete shape type,
which reveals the underlying function type
and enable the application to be typechecked.

\subsection{\ccAdp{}: Term-Level Box Inference}

\begin{wide-rules}

\textbf{Typing \quad $\typAlgoDft{t}{t^\prime}{T}$}


  \infrule[\ruledef{bi-var}]
  {x : \tCap{C}{S} \in \G}
  {\typAlgoDft{x}{x}{\tCap{C}{S}}}

  \infrule[\ruledef{bi-tabs}]
  {\typAlgo{(\G, X <: S)}{t}{t^\prime}{T} \\ \wfTyp{\G}{S}}
  {\typAlgoDft{\tTLambda{X}{S}{t}}{\tTLambda{X}{S}{t^\prime}}{\cv{t^\prime}\,\tTForall{X}{S}{T}}}

  \infrule[\ruledef{bi-abs}]
  {\typAlgo{\extendG{x}{U}}{t}{t^\prime}{T} \\ \wfTyp{\G}{U}}
  {\typAlgoDft{\tLambda{x}{U}{t}}{\tLambda{x}{U}{t^\prime}}{\cv{t^\prime}/x \ \tForall{x}{U}{T}}}


  \infrule[\ruledef{bi-app}]
  {\typUbDft{x}{t_x}{C\,\tForall{z}{U}{T}}\\
    \new{\adpTypDft{y}{t_y}{U}}\\
    t = \fEmbed{ \tLet{x^\prime}{t_x}{\tLet{y^\prime}{t_y}{x^\prime\,y^\prime}} }\\
    T^\prime = \begin{cases}
      [z:=y]T, &\text{if $t_y$ is a variable} \\
      \avoidOp{y^\prime}{\cv{U}}{[z:=y^\prime]T}, &\text{otherwise} \\
    \end{cases}
    }
  {\typAlgoDft{x\,y}{t}{T^\prime}}

  \infrule[\ruledef{bi-tapp}]
  {\typUbDft{x}{t_x}{\tCap{C}{\tTForall X S T}} \\
    \subDft{S^\prime}{S}}
  {\typAlgoDft{x[S^\prime]}{\tLet{x^\prime}{t_x}{x^\prime [S^\prime]}}{[X := S']T}}







\begin{multicols}{2}

  \infrule[\ruledef{bi-box}]
  {\typUpDft{x}{\tCap{C}{S}} \andalso C \subseteq \dom{\G}}
  {\typAlgoDft{\tBox x}{\tBox x}{\tBox{\tCap C S}}}

  \infrule[\ruledef{bi-unbox}]
  {\typUpDft{x}{\tBox{\tCap C S}} \andalso C \subseteq \dom{\G} \\
    \subDft{C}{C^\prime}}
  {\typAlgoDft{\tUnbox{C^\prime}{x}}{\tUnbox{C^\prime}{x}}{\tCap C S}}

\end{multicols}

\infrule[\ruledef{bi-let}]
{\typAlgoDft{s}{s^\prime}{U} \\
  \typAlgo{\G, x: U}{t}{t^\prime}{T} \\
  T^\prime = \avoidOp{x}{\cv{U}}{T}}
{\typAlgoDft{\tLet{x}{s}{t}}{\tLet{x}{s^\prime}{t^\prime}}{T^\prime}}


\textbf{Variable Typing with Unboxing \quad $\typUbDft{x}{t}{T}$}

\infrule[\ruledef{var-return}]
  {\typUpDft{x}{T}}
  {\typUbDft{x}{x}{T}}

\infrule[\ruledef{var-unbox}]
  {\typUpDft{x}{\tBox{\tCap{C}{S}}} \andalso C \subseteq \dom{\G}}
  {\typUbDft{x}{\tUnbox{C}{x}}{\tCap{C}{S}}}

\caption{Typing rules of System \ccAdp{}.}
  \label{fig:adp-typ}

\end{wide-rules}

\begin{wide-rules}

\textbf{Box Adaptation \quad $\adpTypDft{x}{T}{U}$}

\infrule[\ruledef{adapt}]
{\typAlgDft{x}{T_0} \\
  \adpDft{x}{T_0}{t}{T} \\
  \wfTyp{\G}{T}
}
{\adpTypDft{x}{t}{T}}

\textbf{Adaptation subtyping \quad $\adpDft{x}{T}{t_x}{U}$}

\begin{multicols}{2}

\infrule[\ruledef{ba-refl}]
{\subDft{C}{C^\prime}}
{\adpDft{x}{\tCap{C}{S}}{x}{\tCap{C^\prime}{S}}}

\infrule[\ruledef{ba-tvar}]
{X <: S \in \G \\
  \adpDft{x}{\tCap{C}{S}}{t_x}{\tCap{C^\prime}{S^\prime}}}
{\adpDft{x}{\tCap{C}{X}}{t_x}{\tCap{C^\prime}{S^\prime}}}

\infrule[\ruledef{ba-top}]
{\subDft{C}{C^\prime}}
{\adpDft{x}{\tCap{C}{S}}{x}{\tCap{C^\prime}{\top}}}

\infrule[\ruledef{ba-boxed}]
{
  \adp{\G}{y}{\tCap{C_1}{S_1}}{t_y}{\tCap{C_2}{S_2}}\\
  t = \tLet{y}{\tUnbox{C_1}{x}}{\tLet{z}{t_y}{\tBox{z}}}}
{\adpDft{x}{C_x\,\tBox{\tCap{C_1}{S_1}}}{\fEmbed t}{\tCap{C_x^{\prime}}{\tBox{\tCap{C_2}{S_2}}}}}





\end{multicols}

\infrule[\ruledef{ba-fun}]
{\adp{\G}{x}{U_2}{t_x}{U_1} \\
  \adp{(\G, x: U_2, x' : U_1)}{z}{T'_1}{t_z}{T_2}\\
  T'_1 =
  \begin{cases}
    T_1, &\text{if $t_x = x$,} \\
    \fsubst {x} {x'} T_1, &\text{otherwise,} \\
  \end{cases} \\
  t_f = {\tLambda{x}{U_2}{\tLet{x^\prime}{t_x}{\tLet{z}{x_f\ x^\prime}{t_z}}}} \\
  \subDft{\substIn{x_f}{C}{\cv{\fEmbed{t_f}}}}{C^\prime}
}
{\adpDft{x_f}{C\,\tForall{x}{U_1}{T_1}}{\fEmbed{t_f}}{C^\prime\,\tForall{x}{U_2}{T_2}}}

\infrule[\ruledef{ba-tfun}]
{
  \subDft{S_2}{S_1} \\
  \adp{(\G, X <: S_2)}{z}{T_1}{t_z}{T_2} \\
  t_f = \tTLambda{X}{S_2}{\tLet{z}{x_f[X]}{t_z}} \\
  \subDft{\substIn{x_f}{C}{\cv{\fEmbed{t_f}}}}{C^\prime}
}
{\adpDft{x_f}{C\,\tTForall{X}{S_1}{T_1}}{\fEmbed{t_f}}{C^\prime\,\tTForall{X}{S_2}{T_2}}}


\infrule[\ruledef{ba-box}]
{\adpDft{x}{\tCap{C}{S}}{t_x}{\tCap{C^\prime}{S^\prime}} \andalso
  {C^\prime} \subseteq \dom{\G}}
{\adpDft{x}{\tCap{C}{S}}{\tLet{y}{t_x}{\tBox{y}}}{\tCap{C^{\prime\prime}}{\tBox{\tCap{C^\prime}{S^\prime}}}}}

\infrule[\ruledef{ba-unbox}]
{
  \adp{\G}{y}{\tCap C S}{t_y}{\tCap{C^\prime}{S^\prime}} \andalso
  C \subseteq \dom{\G}
}
{\adpDft{x}{\tBox{\tCap C S}}{\tLet{y}{\tUnbox{C}{x}}{t_y}}{\tCap{C^\prime}{S^\prime}}}


\caption{Box adaptation rules of System \ccAdp{}.}
  \label{fig:adp-sub}

\end{wide-rules}

Based on the semi-algorithmic type system \systemAlgo{},
we develop the term-level box inference system \ccAdp{}.
Its typing and subtyping rules are defined in Figure \ref{fig:adp-typ} and \ref{fig:adp-sub} respectively.
The term syntax is unchanged.
The typing and subtyping rules are based on the semi-algorithmic system,
but equipped with the funcationality to inference the missing boxes by transforming the input term.

\subsubsection{Typing}

The typing judgement in \ccAdp{} is now in the form of $\typAlgoDft{t}{t^\prime}{T}$.
It could be interpreted as the statement that the input term $t$,
which is possibly ill-typed with missing boxes,
can be transformed to $t^\prime$ by box inference.
The result term $t^\prime$ has the missed box operations completed,
and is well-typed as $T$ in \cc{},
which we will show formally in metatheory.
The typing rules in \ccAdp{} looks very much like the ones in \ccAlgo{}.
The main differences are (1) \ccAdp{} uses box adaptation instead of subtype check in \textsc{app} rule,
and (2) \ccAdp{}'s typing rules keep track of the transformed term $t^\prime$ in the judgement.

Box adaptation is a replacement for subtyping, which inserts box operations in addition to regular subtype checks to make the adapted term conforming to the expected type.
The details of box adaptation will be given in Section \ref{sec:box-adaptation}.
Another place where box inference is involved is the variable typing judgement $\typUbDft xTt$.
In addition to widening the type variables as in $\typUpDft xT$,
variable typing with unboxing inserts an unbox operation, if the widened type is boxed.
We use this rule to lookup variables in \adprn{app} and \adprn{tapp},
because this two rules expect the variable $x$ to be functions,
and $x$'s type may be a boxed function after widening the type variables.
In such cases, we perform box inference to unbox the variable, so that the system is more expressive.

In \adprn{app}, we return the normalized term.
Term normalization $\fEmbed{t}$ applies simplifications on the term.
It will be defined in Section \ref{sec:box-adaptation},
and among the simplifications what matters here is that
the let-bindings in the form of $\tLet{x}{y}{t}$ will be reduced into $t[x \mapsto y]$.
This is necessary for the completeness of \ccAdp{} (i.e. any term typeable in \ccAlgo{} should be typeable in \ccAdp{}).
Now we demonstrate its necessity.
In \adprn{app} we have to bind the adapted functions and arguments before applying them,
due to the monadic normal form the calculus is in.
We will re-bind the variables $x$ and $y$ even if box adaptation keeps them untouched,
which implies that $x$ and $y$ are well-typed in the original system without box inference involved.
In this case,
the transformed term before normalization will be $\tLet{x^\prime}{x}{\tLet{y^\prime}{y}{x^\prime\,y^\prime}}$.
Instead of being typed as $[z := y]T$ (which is derivable in \ccAlgo{}),
it will be typed as $\avoidOp{y^\prime}{\set{y}}{[z := y^\prime]T}$,
which is larger than $[z := y]T$ since it avoids $y^\prime$ to the empty set at contravariant places.
To ensure completeness, we have to simplify $\tLet{x^\prime}{x}{\tLet{y^\prime}{y}{x^\prime\,y^\prime}}$ to
$x\,y$ so that the most precise type can be derived.

The \adprn{var} rule looks up the variable without transforming it.
The \adprn{abs} and \adprn{tabs} rule types their body with box inference,
transforming the lambda to a new one with the adapted body.
The capture set of the closure are computed based on the adapted body, too.
The \adprn{tapp} rule unboxes the function variable when possible,
and bind it to a variable in order to apply it.
Regular subtype check is used to compare the argument type and the bound.
The \adprn{box}, \adprn{unbox} and \adprn{let} rule types the term
as in the semi-algorithmic system, returning the input term unchanged.

\subsubsection{Box Adaptation}
\label{sec:box-adaptation}

The box adaptation judgement $\adpTypDft{x}{t_x}{T}$ (defined in Figure \ref{fig:adp-sub})
states that the variable $x$ can be box-adapted to term $t_x$ so that $t_x$ conforms to the expected type $T$.
It first looks up the type of the variable $x$ in the environment,
and invokes the \emph{adaptation subtyping} to perform box inference.

The adaptation subtyping $\adpDft{x}{T}{t_x}{U}$ could be understood as:
given a variable $x$ of type $T$, it can be box-adapted into term $t_x$ so that it conforms to the expected type $U$.
The adaptation subtyping is developed on the basis of semi-algorithmic subtyping defined in \ccAlgo{}.
Compared to semi-algorithmic subtyping, adaptation subtyping
(1) transforms the input variable $x$ and returns the adapted term $t_x$ which can be typed as $U$ given that $x$ is bound to $T$ in the typing context,
and (2) has the \adprn{box} and \adprn{unbox} rules for inserting box- and unbox-operations
when there is a mismatch of boxes between the actual type and the expected type.

\textbf{Inlining the \algorn{capt} rule.}
In System \ccAdp{} the \algorn{capt} rule has been inlined to each of the adaptation subtyping rules.
This is because box adaptation could change the captured variables of the term, which could impose additional constraints on the capture sets of closures' types.
For instance, in the \adprn{unbox} rule,
the set of captured variables changes from $\cv{x} = \set{x}$ to $\cv{\tLet{y}{\tUnbox{C}{x}}{t_y}} = \set{x} \cup C \cup \cv{t_y}$.
In the \adprn{fun} rule, since $x_f$ is adapted into $t_f^\prime$,
which possibly captures more variables (e.g. due to an unbox operation inserted by \adprn{unbox} when adapting the function body),
we have to make sure that $C^\prime$ subcaptures the captured variables of the adapted term in the premise.

\textbf{Eta-expansion.}
In \ruleref{ba-fun} and \ruleref{ba-tfun} rules, we eta-expand the input variable, so that we can box-adapt the argument and the result.
For instance, in the \ruleref{ba-fun} rule, the general idea can be illustrated in the following diagram.
\begin{equation*}
  \begin{split}
    x_f \xrightarrow{\text{eta-expansion}} &\quad\tLambda{x}{U_2}{x_f\,x} \\
    \xrightarrow[\text{argument}]{\text{adapt}} &\quad\tLambda{x}{U_2}{\tLet{x^\prime}{t_x}{x_f\,x^\prime}} \\
    \xrightarrow[\text{result}]{\text{adapt}} &\quad\tLambda{x}{U_2}{\tLet{x^\prime}{t_x}{\tLet{z}{x_f\,x^\prime}{t_z}}} \\
  \end{split}
\end{equation*}
Since box adaptation assumes that the input is a variable,
we let-bind the adapted argument, and also the result of the application.
Note that when adapting the result of the application in the \ruleref{ba-fun} rule,
we substitute the dependent reference to $x$ in the result type to the adapted argument $x'$,
which is in line with the fact that the function is applied to $x'$.
However, the substitution is omitted when the adapted argument is a variable.
This is the outcome of term normalisation, which we will delve into soon.

\textbf{Adaptation inside boxes.}
The \ruleref{boxed} rule in \cc{} derives subtyping relations between boxed types by comparing the types inside the box.
In the box inference system, the \adprn{boxed} rule adapts the type inside the box.
To do this, we first unbox the term, let-bind it, box-adapt the unboxed variable, and finally box it.

\textbf{Term normalisation.}
In the rules (\ruleref{ba-fun}, \ruleref{ba-tfun} and \ruleref{ba-boxed}),
$\fEmbed{t}$ denotes the normalisation of a term $t$.
The normalisation inlines let-bindings, beta-reduces functions, and simplifies box operations,
which is defined formally as follows.
\begin{definition}[Term normalisation]
  The normalisation of a term $t$, written $\fEmbed{t}$, is defined by the following equations.
  \begin{align*}
    &\fEmbed{\tLet{x}{s}{t}} \quad &= \quad &\compactLet{\tLet{x}{\fEmbed{s}}{\fEmbed{t}}}, & \label{norm-let}\tag{\textsc{norm-let}} \\
    &\fEmbed{\tLambda{x}{T}{u}} \quad &= \quad &\compactLam{\tLambda{x}{T}{\fEmbed{u}}}, & \label{norm-fun}\tag{\textsc{norm-fun}} \\
    &\fEmbed{\tTLambda{X}{S}{u}} \quad &= \quad &\compactLam{\tTLambda{X}{S}{\fEmbed{u}}}, &&\label{norm-tfun}\tag{\textsc{norm-tfun}} \\
    &\fEmbed{t} \quad &= \quad &t,  \quad&\textrm{otherwise}, \label{norm-return}\tag{\textsc{norm-return}} \\
  \end{align*}
  where $\compactLet{\cdot}$ and $\compactLam{\cdot}$ are helper functions that
  normalise the let-bindings and performs the beta-reductions respectively.
  They are given as follows.
  \begin{align*}
    &\compactLet{\tLet{y}{\tUnbox{C}{x}}{\tBox{y}}} \quad &= \quad &x, &\label{nl-box}\tag{\textsc{nl-box}} \\
    &\compactLet{\tLet{y}{t}{y}} \quad &= \quad &t, &\label{nl-deref}\tag{\textsc{nl-deref}} \\
    &\compactLet{\tLet{y}{x}{t}} \quad &= \quad &t[y \mapsto x]. &\label{nl-rename}\tag{\textsc{nl-rename}} \\
    & \\
    &\compactLam{\tLambda{z}{T}{x\,z}} \quad &= \quad &x, & \label{nl-beta}\tag{\textsc{nl-beta}} \\
    &\compactLam{\tTLambda{X}{S}{x[X]}} \quad &= \quad &x. & \label{nl-tbeta}\tag{\textsc{nl-tbeta}} \\
  \end{align*}
\end{definition}
Term normalisation is essential for the completeness of box adaptation system.
In metatheory, to show the completeness of the box inference system,
we have to prove that if $\subAlgoDft{T}{U}$ is derivable in \ccAlgo{},
then we have $\adpDft{x}{T}{x}{U}$ in \ccAdp{}.
In other words, if $T$ is a subtype of $U$,
which means that there is no box adaptation needed to turn a variable $x$ of type $T$ into $U$,
adaptation subtyping should return the input variable as it is.
However, due to the MNF nature of capture calculus systems,
and the fact that box adaptation works only on variables as input,
we insert let bindings everywhere during the adaptation.
Plus, the function gets eta-expanded in function adaptation rules.
Therefore, without term normalisation, the box adaptation system transforms the term even if there is a subtyping relation between the actual and the expected type,
which will undermine the completeness of box inference,
as we will see in the metatheory.

\subsection{\ccAdpt{}: Type-Level Box Inference}

\newcommand{\thisx}{\ensuremath{\lozenge}}

\begin{wide-rules}



\textbf{Box Adaptation \quad $\adptTypDft{x}{U}{\cCat}{C}$}

  \infrule[\ruledef{t-adapt}]
  {\typAlgDft{x}{U} \\
  \adptcDft{U}{T}{\cCat}{{C}} \andalso \wfTyp{\G}{T}
  }
  {\adptTypDft{x}{T}{\cCat}{C[x]}}



\textbf{Adaptation subtyping \quad $\adptcDft{T}{U}{\cCat}{C}$}

\begin{multicols}{3}

\infrule[\ruledef{t-ba-refl}]
  {\subDft{C}{C^\prime}}
  {\adptcDft{\tCap{C}{S}}{\tCap{C^\prime}{S}}{\cVar}{\set\thisx}}

\infrule[\ruledef{t-ba-tvar}]
  {X <: S \in \G \\
  \adptcDft{\tCap{C}{S}}{\tCap{C^\prime}{S^\prime}}{\cCat}{C}}
  {\adptcDft{\tCap{C}{X}}{\tCap{C^\prime}{S^\prime}}{\cCat}{C}}

\infrule[\ruledef{t-ba-top}]
  {\subDft{C}{C^\prime}}
  {\adptcDft{\tCap{C}{S}}{\tCap{C^\prime}{\top}}{\cVar}{\set{\thisx}}}

\end{multicols}

\infrule[\ruledef{t-ba-boxed}]
  {\adptcDft{\tCap{C_1}{S_1}}{\tCap{C_2}{S_2}}{\cCat}{C}\\
   (\cCat^\prime, C^\prime) =
   \begin{cases}
     (\cVar, \set\thisx), &\text{if $\cCat$ is $\cVar$} \\
     (\cTrm, C_1 \cup \set\thisx), &\text{if $\cCat$ is $\cVal$} \\
     (\cTrm, C \cup C_1 \cup \set\thisx), &\text{otherwise} \\
   \end{cases}
   }
  {\adptcDft{C_x\,\tBox{\tCap{C_1}{S_1}}}{\tCap{C_y}{\tBox{\tCap{C_2}{S_2}}}}{\cCat^\prime}{C^\prime}}





\infrule[\ruledef{t-ba-fun}]
{\adptcDft{U_2}{U_1}{\cCat_1}{C_{1}} \\
  \adptc{\G, x: U_2, x': U_1}{T'_1}{T_2}{\cCat_2}{C_{2}}\\
  C_f = C_1 \cup C_2 \setminus \set{x, x'} \cup \set{\thisx} \andalso
  \subDft{C_f[C]}{C^\prime} \\
  T'_1 =
  \begin{cases}
    T_1, &\text{if $\cCat = \cVar$} \\
    \fsubst{x}{x'} T_1, &\text{otherwise} \\
  \end{cases}\andalso
  \cCat^\prime = \begin{cases}
    \cVar, &\text{if both $\cCat_1$ and $\cCat_2$ are $\cVar$} \\
    \cVal, &\text{otherwise} \\
  \end{cases}
}
{\adptcDft{C\,\tForall{x}{U_1}{T_1}}{C^\prime\,\tForall{x}{U_2}{T_2}}{\cCat^\prime}{C_f}}

\infrule[\ruledef{t-ba-tfun}]
{
  \subAlgoDft{S_2}{S_1} \\
  \adptc{\G, X <: S_2}{T_1}{T_2}{\cCat}{C_2} \\
  C_f = C_2\setminus\set{\thisx}\andalso
  \subDft{C_f[C]}{C^\prime} \\
  \cCat^\prime = \begin{cases}
    \cVar, &\text{if $\cCat$ is $\cVar$} \\
    \cVal, &\text{otherwise} \\
  \end{cases}
}
{\adptcDft{C\,\tTForall{X}{S_1}{T_1}}{C^\prime\,\tTForall{X}{S_2}{T_2}}{\cCat^\prime}{C_f}}

\begin{multicols}{2}

\infrule[\ruledef{t-ba-box}]
  {\adptcDft{\tCap{C_1}{S_1}}{\tCap{C_2}{S_2}}{\cCat}{C_0} \\
    {C_2} \subseteq \dom{\G} \\
  C^\prime = \begin{cases}
    \eset, &\text{if $\cCat \in \set{\cVar, \cVal}$} \\
    C_0,   &\text{otherwise}
  \end{cases}
  }
  {\adptcDft{\tCap{C_1}{S_1}}{\tCap{D}{\tBox{\tCap{C_2}{S_2}}}}{\cTrm}{{C^\prime}}}

\infrule[\ruledef{t-ba-unbox}]
{
  \adptcDft{\tCap {C_1} {S_1}}{\tCap{C_2}{S_2}}{\cCat}{C_0} \andalso
  {C_1} \subseteq \dom{\G}
}
{\adptcDft{C_x\,\tBox{\tCap{C_1}{C_2}}}{\tCap{C_2}{S_2}}{\cTrm}{C_1 \cup C_0 \cup \set\thisx}}

\end{multicols}

\caption{Box adaptation rules of System \ccAdpt{}.}
  \label{fig:adpt-sub}

\end{wide-rules}

\begin{wide-rules}

\textbf{Typing \quad $\typAdptDft{t}{T}{C}$}

\begin{multicols}{2}

  \infrule[\ruledef{t-bi-var}]
  {x : \tCap C S \in \G}
  {\typAdptDft{x}{\tCap{\set{x}}{S}}{\set{x}}}

  \infrule[\ruledef{t-bi-tabs}]
  {\typAdpt{(G, X <: S)}{t}{T}{C} \andalso \wfTyp{\G}{S}}
  {\typAdptDft{\tTLambda{X}{S}{t}}{C\,\tTForall{X}{S}{T}}{C}}

  \infrule[\ruledef{t-bi-abs}]
  {\typAdpt{\extendG{x}{U}}{t}{T}{C} \andalso \wfTyp{\G}{U}}
  {\typAdptDft{\tLambda{x}{U}{t}}{C/x \ \tForall{x}{U}{T}}{C/x}}



  \infrule[\ruledef{t-bi-app}]
  {\typUbtDft{x}{C\,\tForall{z}{U}{T}}{C_{1}}\\
  {\adptTypDft{y}{U}{\cCat}{C_{2}}} \\
  T^\prime = \begin{cases}
    [z:=y]T, &\text{if $\cCat = \cVar$}\\
    \avoidOp{z}{\cv{U}}{T}, &\text{otherwise} \\
  \end{cases}
  }
  {\typAdptDft{x\,y}{T^\prime}{C_{1} \cup C_{2}}}




  \infrule[\ruledef{t-bi-tapp}]
  {\typUbtDft{x}{\tCap{C}{\tTForall X S T}}{C} \\
    \subDft{S^\prime}{S}}
  {\typAdptDft{x[S^\prime]}{[X := S']T}{C}}


\end{multicols}

\begin{multicols}{2}

  \infrule[\ruledef{t-bi-box}]
  {\typUpDft{x}{\tCap{C}{S}} \andalso C \subseteq \dom{\G}}
  {\typAdptDft{\tBox x}{\tBox{\tCap C S}}{\eset{}}}

  \infrule[\ruledef{t-bi-unbox}]
  {\typUpDft{x}{\tBox{\tCap C S}} \andalso C \subseteq \dom{\G} \\
    \subDft{C}{C^\prime}}
  {\typAdptDft{\tUnbox{C^\prime}{x}}{\tCap C S}{C^\prime \cup \set{x}}}

\end{multicols}

\infrule[\ruledef{t-bi-let}]
{\typAdptDft{s}{U}{C_{1}} \\
  \typAdpt{(\G, x: U)}{t}{T}{C_{2}} \\
  T^\prime = \avoidOp{x}{\cv{U}}{T} \\
  C^\prime = \begin{cases}
    C_{2}, &\text{if $s$ is a value and $x \notin C_{2}$} \\
    C_{1} \cup C_{2} / \set{x}, &\text{otherwise} \\
  \end{cases}
  }
{\typAdptDft{\tLet{x}{s}{t}}{T^\prime}{C^\prime}}


\textbf{Variable Typing with Unboxing \quad $\typUbtDft{x}{T}{C}$}

\infrule[\ruledef{t-var-return}]
  {\typUpDft{x}{T}}
  {\typUbtDft{x}{T}{\set{x}}}
\infrule[\ruledef{t-var-unbox}]
  {\typUpDft{x}{\tBox{\tCap{C}{S}}} \andalso C \subseteq \dom{\G}}
  {\typUbtDft{x}{\tCap{C}{S}}{C \cup \set{x}}}

%
%

\caption{Typing rules of System \ccAdpt{}.}
  \label{fig:adpt-typ}

\end{wide-rules}

In term-level box inference system, the transformation on the terms (i.e.
inserting box-operations, let-bindings and eta-expansions)
does not change the semantics of the program.
What matters is \emph{whether} it is {possible} to transform the program to be well-typed with box inference,
but not \emph{what} is the result term of box inference.
To eliminate unnecessary computational and memory burden
to compute and store the result terms,
in this section we investigate the possibility of \emph{predicting} the effect of box inference solely on the type level.

In this section, on the basis of the term-level box inference system,
we further develop the type-level system \ccAdpt{}, which does the equivalent reasoning without computing the transformed terms.
Figure \ref{fig:adpt-typ} and \ref{fig:adpt-sub} shows
the typing and adaptation rules of \ccAdpt{}.

The key problem of modeling box inference on type level is how to predict the effects of term transformations.
First of all, the transformation inserts boxes and unboxes, thus adding or dropping boxes in types.
More importantly, we have to predict how the term transformations changes the capture sets.
It turns out that to do this we will need to keep track of the captured variable sets of transformed terms.

\subsubsection{Type-Level Adaptation Subtyping}

The type-level adaptation subtyping rules has the form $\adptcDft{T}{U}{\cCat}{C}$.
In the metatheory, we prove that
if $\adptcDft{T}{U}{\cCat}{C}$ is derivable in \ccAdpt{},
we have $\adpDft{x}{T}{t_x}{U}$ in the term-level system \ccAdp{},
such that $C$ is the captured variables of $t_x$ (i.e. $\cv{t_x}$).
In other words, the type-level system models the equivalent transformation,
and predict the captured variables of the adapted term,
without explicitly computing the term.
It turns out that to predict the captured variable set correctly,
we have to keep track of the \emph{kind} of terms too, which is the $\cCat$ in the judgement,
and we have $\fCat{t_x} = \cCat$, where $\fCat{t_x}$ returns $t_x$'s kind.

\textbf{Term kinds.}
We categorize the terms into three kinds, which are
\begin{itemize}
\item \textbf{Variables $\cVar$};
\item \textbf{Values $\cVal$}: lambda $\tLambda{x}{T}{t}$, type lambda $\tTLambda{X}{S}{t}$, and box $\tBox{x}$;
\item \textbf{Terms $\cTrm$}: application $x\,y$, type application $x[S]$, let binding $\tLet{x}{t}{u}$, and unbox $\tUnbox{C}{x}$.
\end{itemize}
We observe that it is necessary to keep track of the kind of the adapted terms
so that we can correctly predict the captured variables of them.
For example, in the \adptrn{box} rule,
given the input variable $x$, its actual type $\tCap{C}{S}$ and the expected type $\tBox{\tCap{C^\prime}{S^\prime}}$,
the \adprn{box} rule in \ccAdp{} first adapts the variable to $\tCap{C^\prime}{S^\prime}$ by transforming it into $t_x$,
then boxes it, resulting in the term $\tLet{y}{t_x}{\tBox{y}}$.
Based on the deinition of $\cv{\cdot}$, if $t_x$ is either a variable or a value (i.e. $\fCat{t_x} \in {\cVar, \cVal}$),
$\cv{\tLet{y}{t_x}{\tBox{y}}} = \set{}$, otherwise the captured variable set will be $\cv{t_x}$.
In the \adptrn{box} rule in System \ccAdpt{}, we make the correct prediction of the captured variables based on the kind of the adapted term.
This illustrates the necessity of keeping track of the \emph{kind} during type-level adaptation.


\textbf{Holes in Adaptation Subtyping.}
In the adaptation subtyping rules we have a special construct, hole, written $\thisx$.
We need this special \emph{hole} construct, because
the captured variable sets predicted by the system
could possibly contains the input variable $x$,
but $x$ is unknown in the adaptation subtyping derivations.
We therefore use $\thisx$ as a placeholder for this input variable,
and fill the hole with the actual $x$ when it becomes available, as in \ruleref{t-adapt}.
$C[x]$ fills the hold in $C$ with $x$,
i.e. $C[x] = C \setminus \set{\thisx} \cup \set{x}$.


\textbf{Adaptation rules.}
Now we inspect these adaptation subtyping rules one-by-one.
These rules are formed in a way that each of them predicts the captured variable sets and the term kinds in the corresponding rule of the term-level system.
The \ruleref{t-ba-refl} and \ruleref{t-ba-top} rule is straightforward: on the term level this two rules returns the input variable as it is.
Therefore we predict that the kind of the adapted term is $\cVar$ and the captured variables is just $\set{\thisx}$.
The \ruleref{t-ba-tvar} rule widens the type variable to its bound and calls the rule recursively, which reflects what happens on the term level.
The \ruleref{t-ba-boxed} rule performs a case analysis on $\cCat$ (which corresponds to the kind of $t_y$ in \ruleref{t-ba-boxed}).
If $\cCat$ is the variables, on term-level the $\compact{\cdot}$ simplifies the result term to the input variable,
so the result kind is $\cVar$ and the captured variable set is just $\set{\thisx}$.
If $\cCat$ is $\cVal$ or $\cTrm$, the let-bindings and box/unbox operators will not get simplified.
So the result kind is $\cTrm$.
Both $C_1$ and $\set{\thisx}$ are in the captured variable set, because of the unbox operation $\tUnbox{C}{x}$ in the result term.
$C$ is not included in the captured variable set if $t_y$ is a value.
In the \ruleref{t-ba-fun} rule,
we adapt the argument and the result recursively in the premise.
The result kind is $\cVar$ only if both $\cCat_1$ and $\cCat_2$ are $\cVar$
(which means the resulted term is simplified to the input variable $x_f$ by $\compact{\cdot}$),
and it is otherwise always $\cVal$ since the adapted term will be a lambda.
We predict the captured variables set with $C_1$ and $C_2$.
The \ruleref{t-ba-tfun} is really similar.
The resulted kind is $\cVar$ when $\cCat$ is $\cVar$, in which case the result is normalised to the variable.
The captured variable set could again be computed from $C_2$.
Note that regular subtyping check $\subAlgoDft{S_2}{S_1}$ is used to compare the bounds.
For the \adptrn{box} and \adptrn{unbox} rule, the resulted kind is always $\cTrm$,
since \adprn{box} or \adprn{unbox} always return a let expression.
The captured variable set can be computed using $C_0$.
Notably, in the \adprn{box} rule the result term captures the empty set
when $\cCat$ is $\cVar$ or $\cVal$.
In this case the related references, including the input variable, are hidden by the box operation.


\textbf{Boxes on the top.}
After inspecting the adaptation subtyping rules in System \ccAdpt{},
we observe that the hole $\thisx$ is always an element in the captrued variable set,
unless the \adptrn{box} is in the root of the derivation tree,
in which case a box is inferred on the top of the term.
In other words,
box adaptation only \emph{hides} references in this specific case;
other than this,
box adaptation always makes the term to capture \emph{more} references.
This observation facilitates the implementation of box inference.


\subsubsection{Typing}

The typing judgement $\typAdptDft{t}{T}{C}$ in System \ccAdpt{}
now returns the captured variable set of the adapted term,
without actually computes the term.
In metatheory, it is proven that if $\typAdptDft{t}{T}{C}$,
in the term-level system, the typing judgement $\typAlgoDft{t}{t^\prime}{T}$
is derivable for some term $t^\prime$,
and $\cv{t^\prime} = C$.
In other words, the type-level typing rule predicts
\emph{whether} the input program $t$ can become well-typed after box inference,
and returns its captured variable set if so.

In System \ccAdpt{},
the captured variable sets of resulted terms are predicted by the judgement,
but the actual terms resulted from box are unknown.
The \adptrn{var}, \adptrn{box}, \adptrn{tapp} and \adptrn{unbox} rules are standard
and expected, while they returns the captured variable set of the transformed term
instead of the term itself.
\ccAdpt{} uses the type-level box adaptation in \adptrn{app}.
In the \adptrn{abs} and \adptrn{tabs} we make use of these predicted captured variable sets,
instead of computing them using $\cv{\cdot}$ as in the term-level system.
In the \adptrn{let} rule,
we do a case analysis to compute the correct captured set variable,
which reflects the two cases for let-expressions in the definition of $\cv{\cdot}$.
In the same vein as the typing rules,
the variable typing rule now returns the captured set of the possibly unboxed term,
instead of returning this term directly.

\section{Metatheory}

\newcommand{\wtp}[1]{\mathsf{wtp}({#1})}

In the metatheory, we develop the proof of the relations between
System \cc{} and the family of algorithmic box inference systems.
The relationships between these systems is illustrated in the following diagram.
\begin{equation*}
  \wtp{\text{\cc{}}} \quad
  = \quad \wtp{\text{\ccAlgo{}}} \quad
  \subset \quad \wtp{\text{\ccAdp{}}} \quad
  = \quad \wtp{\text{\ccAdpt{}}}
\end{equation*}
Here $\wtp{\cdot}$ denotes the set of well-typed programs in the system.
As depicted in the diagram, an equivalence exists between \cc{} and \ccAlgo{},
as well as between \ccAdp{} and \ccAdpt{}.
However, System \ccAdp{} is \emph{more expressive} than System \ccAlgo{},
as it accepts more programs
by resolving the inconsistencies between boxes with box infernece.

\subsection{Equivalence Between \cc{} and \ccAlgo{}}

We first prove the equivalence between \cc{} and \ccAlgo{}.
It implies that a correspondance exists between every derivation in \cc{} and \ccAlgo{},
This is expressed in the the following theorem.
\begin{theorem}[Equivalence between \cc{} and \ccAlgo{}]
	$\forall \G, t, T. \typDft{t}{T} \leftrightarrow \typAlgDft{t}{T}$.
\end{theorem}
This theorem follows directly from Theorem \ref{thm:typ-algo-completeness} and \ref{thm:typ-algo-soundness}.
Detailed proof is given in Section \ref{sec:cc-algo-proof}.

The proof of this theorem replies on the admissibility of
both the reflexivity and the transitivity rule in \ccAlgo{}'s subtyping system.
The following two lemmas demonstrates the admissibility of the rules respectively.
\begin{restatable}[Reflexivity of algorithmic subtyping]{lemma}{subalgorefl}  \label{lemma:sub-algo-refl}
  For any $T$, $\subAlgoDft{T}{T}$.
\end{restatable}
This lemma's proof can be easily established through induction on the type $T$.

\begin{restatable}[Transitivity of algorithmic subtyping $\subAlgoDft{T_1}{T_2}$]{lemma}{subalgotrans} \label{lemma:algorithmic-subtyping-transitivity}
  In any environment $\G$, $\subAlgoDft{T_1}{T_2}$ and $\subAlgoDft{T_2}{T_3}$, then $\subAlgoDft{T_1}{T_3}$.
\end{restatable}
The proof for this lemma requires induction on the type in the middle (i.e. $T_2$).

In addition to proving the aforementioned two lemmas for subtyping,
it is required to establish a lemma about \emph{avoidance}
which demonstrates that
it is possible to construct the least supertype of a type that avoids a local variable.
This lemma is crucial to proving the case for let-bindings in the completeness theorem.
\begin{lemma}
	If $\typAlgDft{s}{T_1}$,
  and $\typAlg{\G, x: T_1}{t}{T_2}$,
  then $T_2^\prime = \avoidOp{x}{\cv{T_1}}{T_2}$ is the least supertype of $T_2$,
  such that $T_2^\prime$ does not mention $x$.
  Here, least super type means for every $U$ such that $\subAlgo{\G, x: T_1}{T_2}{U}$
  and $x \notin \fv{U}$,
  we have $\subAlgoDft{T_2^\prime}{U}$.
\end{lemma}
We can prove this lemma by induction on the subtype derivation $\subAlgo{\G, x: T_1}{T_2}{U}$.
Once we have established these lemmas,
both direction of the equivalence between the two systems can be proven by induction.

\subsection{Relation Between \ccAlgo{} and \ccAdp{}}

\ccAdp{} is more expressive than \ccAdp{}.
Specifically, every typing derivation \ccAlgo{} can be mapped to a valid derivation in \ccAdp{},
whereas there are certain programs that possess valid typing derivations in \ccAdp{} but not in \ccAlgo{}.
Such a program is ill-typed in \ccAlgo{} due to the inconsistent boxes,
but there could exist a typing derivation $\typAlgoDft{t}{t^\prime}{T}$ in \ccAdpt{}
which resolves the mismatches between boxes through box inference
and thus accept the program $t$.
Also, it is crucial for $t^\prime$ to be well-typed in \ccAlgo{},
or else \ccAdp{} will be unsound.

The following theorem, based on the above ideas,
establishes the relation between the two systems formally.
\begin{theorem}
	$\forall \G, t, T. $ we have:
  (1) if $\typAlgDft{t}{T}$, $\exists t^\prime$ such that $\typAlgoDft{t}{t^\prime}{T}$;
  and (2) if for some $t^\prime$ we have $\typAlgoDft{t}{t^\prime}{T}$, then $\typAlgDft{t^\prime}{T}$.
\end{theorem}

We prove this theorem using Theorem \ref{thm:typ-adp-completeness} and \ref{thm:typ-adp-soundness},
with detailed proof given in Section \ref{sec:cc-adp-proof}.
The completeness and soundness of box adaptation is crucial to the proof,
which is established by the following two lemmas.
\begin{restatable}[Completeness of box adaptation]{lemma}{subadpcompleteness} \label{lemma:sub-algo-to-adp}
  If $\typUpDft{x}{T}$ and $\subAlgoDft{T}{U}$ then $\adpDft{x}{T}{x}{U}$.
\end{restatable}

\begin{restatable}[Soundness of box adaptation]{lemma}{subadpsoundness}  \label{thm:adp-sub-soundness}
  If $\adpTypDft{x}{t_x}{T}$ then $\typDft{t_x}{T}$.
\end{restatable}

Notably, in the completeness lemma (Lemma \ref{lemma:sub-algo-to-adp}),
we demonstrate that
box adaptation will return the input \emph{as it is}
if $T$ is a subtype of $U$
(which means a variable of $T$ can be used as $U$ without box inference involved).
This is important for the completeness of System \ccAdp{}.
To see why, consider the application $x\,y$, where
$\typPrecDft{x}{C\tForall{z}{U}{T}}$,
$\typPrecDft{y}{U^\prime}$,
and $\subAlgoDft{U^\prime}{U}$.
In \ccAdp{}, we lookup $x$,
and box-adapt the argument $y$ into $t_y$.
The resulting term is $\compactLet{\tLet{x^\prime}{x}{\tLet{y^\prime}{t_y}{x^\prime\,y^\prime}}}$.
We have to ensure that $t_y = y$ in this case.
We have to ensure that $t_y = y$
because otherwise the type of the transformed term will be
$\avoidOp{y^\prime}{\cv{U}}{[z := y^\prime]T}$,
which would be wider than the type $[z := y]T$ in System \ccAlgo{}.

Additionally, we introduce an auxilliary subtyping relation,
which is derived from the subtyping rules of \ccAlgo{}
by inlining the \algorn{capt} rule.
It eases the development of the proof on the relation between
subtyping in \ccAlgo{} and adaptation subtyping,
as the adaptation subtyping rules also
have the \textsc{capt} rule inlined.
Figure \ref{fig:capt-sub} shows the definition of this auxilliary relation.
In the metatheory,
we first establish the equivalence between the auxilliary subtyping and the one in \ccAlgo{},
then develop the relation between the auxiliary subtyping rules and box adaptation subtyping rules.


\subsection{Equivalence Between \ccAdp{} and \ccAdpt{}}

In the next part of metatheory,
we aim to demonstrate the equivalence between System \ccAdp{} and \ccAdpt{},
through following theorem.
\begin{theorem}
	$\forall \G, t, T.$
  we have
  (1) if $\typAlgoDft{t}{t^\prime}{T}$ is derivable for some $t^\prime$,
  then $\typAdptDft{t}{T}{C}$ is derivable, where $C = \cv{t^\prime}$;
  and (2) if $\typAdptDft{t}{T}{C}$ is derivable for some $C$,
  then there exists $t^\prime$ such that
  $\typAlgoDft{t}{t^\prime}{T}$ is valid in \ccAdp{}
  and $\cv{t^\prime} = C$.
\end{theorem}
This theorem can be proven using Theorem \ref{thm:typ-adpt-completeness} and \ref{thm:typ-adpt-soundness}.
It indicates that
box inference can be performed with identical expressiveness on type-level,
without storing and computing the adapted terms.
We can prove both directions of this equivalence through induction on the derivation tree.
For a comprehensive proof,
please refer to Section \ref{sec:cc-adpt-proof}.

\subsection{Termintaion of Box Inference}

The type checking process of system F$_{<:}$, which is the basis of System \cc{}, is known to be undecidable.
It has been demonstrated \cite{fsub1} that
certain subtype query will cause the subtype check to loop.

As System \cc{} is based on System F$_{<:}$,
it encounters the same undecidability problem that emerges from subtyping between type functions.
Despite this, the extensions introduced by \cc{} and \ccAdp{} does not degrade the system's decidability.
In other words, the extensions should not lead to non-termination in more situations than in System F.
To demonstrate this idea formally
(in Theorem \ref{lemma:sub-adp-termination} and \ref{lemma:typ-adp-termination}),
we are going to prove that
the typechecking of a program $t$ in \ccAlgo{} terminates
as long as typechecking the program $t^\prime$ does,
where $t^\prime$ is derived from $t$ by erasing all \cc{}-related constructs.

Now we define the function $\eraseCC{\cdot}$ to erase the capture sets and the boxes.
It traverses the type and drops all \cc{}-specific constructs.
\begin{align*}
    &\eraseCC{C\,S} \quad &= \quad  &\eraseCC{S} \\
    &\eraseCC{\tBox{T}} \quad &= \quad &\eraseCC{T} \\
    &\eraseCC{\tForall{z}{T}{U}} \quad &= \quad &\tForall{z}{\eraseCC{T}}{\eraseCC{U}} \\
    &\eraseCC{\tTForall{X}{S}{T}} \quad &= \quad &\tForall{X}{\eraseCC{S}}{\eraseCC{T}} \\
    &\eraseCC{X} \quad &= \quad &X \\
    &\eraseCC{\top} \quad &= \quad &\top \\
\end{align*}
When applied to a context $\G$,
$\eraseCC{\G}$ is straightfowardly defined as the new context
where $\eraseCC{\cdot}$ are applied to each of the bound types.

The following theorems demonstrate the termination conditions of typechecking in \ccAdp{}.
Here, \textsf{subtype} and \textsf{subtype}$_{\textrm{BI}}$ is the subtype checking procedure
of System F$_{<:}$ and \ccAdp{} respectively.
Similarly, \textsf{typecheck} and \textsf{typecheck}$_{\textrm{BI}}$ are the typechecking procedures.
They establish that
the box adaptation and typechecking procedure of any program $t$ in \ccAdp{}
will terminate,
provided that
the subtyping and typechecking procedure
of the erased program $t^\prime$
terminate in System F$_{<:}$.
The simple-formed assumption rules out the types that has nested boxes
(e.g. $\tBox{\tCap{C}{\tBox{C'}{S}}}$),
which do not make sense in practice.
Section \ref{sec:termination-proof} presents the detailed definitions and the proof.

\begin{restatable}[Conditional termination of box adaptation]{theorem}{subadpterm}  \label{lemma:sub-adp-termination}
  Given a well-formed and simple-formed context $\G$
  and simple-formed types $T, U$ that are well-formed in the environment,
  $\subtypeAdp{\G}{T}{U}$ terminates as long as
  $\subtype{\eraseCC{\G}}{\eraseCC{T}}{\eraseCC{U}}$ does.
\end{restatable}

\begin{restatable}[Conditional termination of typechecking in \ccAdp{}]{theorem}{typadpterm}  \label{lemma:typ-adp-termination}
  Given a well-formed and simple-formed context $\G$
  and the term $t$ that are well-formed in the environment and the type annotations are simple-formed,
  if $\typecheck{\eraseCC{\G}}{\eraseCC{t}}$ terminates,
  $\typecheckAdp{\G}{t}$ terminates too.
\end{restatable}

\section{Implementing Capture Checking and Box Inference}

This section discusses the implementation of capture checking and box inference,
based on our experience with Scala 3,
which has practical support for capture checking
available as an experimental language feature.
Although the syntax-directed box inference calculi
directly give rise to a procedure for box inference,
the challenge lies in its efficient implemention and the integration into an existing language implementation.
To begin, we briefly introduce Scala's capture checking implementation,
and based on its framework, we proceed to discuss the implementation of box inference.

\subsection{Scala's Implementation of Capture Checking}

In Scala 3, capture checking is available as an experimental language feature
users may choose to turn on.
Since we expect capture checking to be an extension to the existing compiler
that can be easily enabled or disabled,
it is infeasible to deeply integrate it into the existing logic of type checking.
Therefore, Scala 3 implements capture checking as a standalone phase after typer.

It turns out that it is possible to do vanilla type checking and capture checking separately and in sequence.
During typing, capture set annotations in types are simply ignored,
with types being derived and checked normally.
Then in the capture checking phase,
we compute captured references of closures,
inference missing capture sets
and check the sets against the annotations.
This phase is where box inference is involved.
Also, capture set inference is a crucial part in the implementation of capture checking,
but is left as future work.

The capture checking phase computes the captured references of a closure
(which corresponds to the $\cv{\cdot}$ computation in \textsc{fun} and \textsc{tfun} rules).
Instead of implementing a $\cv{\cdot}$ function that traverses the body of a closure
each time we assign capture sets to functions,
we track the captured references of a term while checking it.
To do this, a capturing environment is created for each closure,
and the environments of nested closures are chained.
References are pushed in these environments as we check the body of these closures.
When checking the capture set in the type of closures,
we retrieve these captured references from the corresponding environments.

Let us inspect how capture checking behaves in the following example:
\begin{minted}{scala}
  def foo(io: {*} IO) =
    def bar(): Unit =
      io.use()
      fs.use()
      ()
    bar
\end{minted}
The function is capture-checked as \mintinline{scala}{{fs} (io: {*} IO) -> {io, fs} () -> Unit}.
When checking the body of \texttt{bar},
we have two capturing environments created for \texttt{bar} and \texttt{foo} respectively,
and the environment for \texttt{bar} has a pointer to \texttt{foo}, since they are nested.
The references \texttt{io} and \texttt{fs} are pushed into \texttt{bar}'s environment,
after which they are propagated to the nesting environment for \texttt{foo}.
However, \texttt{io} is not included in \texttt{foo}'s environment
because it is an argument of \texttt{foo}.
This is in line with the definition of $\cv{\cdot}$.

Explicit capture annotations imposes constraint on the capture sets,
which is enforced by comparing the derived capture sets and the explicit annotations.
This is where box inference is involved.
We want to check whether the derived type of a term conforms to the expectation
with box-related constructs inferred and completed.
The details of box inference implementation is discussed in the following section.
Apart from box inference, we have to infer the capture sets omitted by the users.
To this end, we create \emph{capture set variables} for the captrue sets to be inferred,
and record the subcapturing relationships as we check the program.
The algorithm for capture set inference is important in the implementation of capture checking,
and gathers considerable complexity to fit in with the existing architecture of the compiler.
The formalization of capture set inference is worth investigation too,
and is left as possible future work.

\subsection{Implementing Box Inference}

\subsubsection{Box Adaptation as a Separate Check}

The central part of box inference implementation is box adaptation.
In the box inference calculi,
box adaptation is proposed as a replacement for subtyping checks,
which simultaneously adapts the mismatched boxes and performs regular subtyping reasoning.
However,
completely rewriting the subtype comparison logic and
reframing it with box adaptation is arguably too heavy a change,
and makes the implementation less plugable than it should be.
It is therefore favorable to implementation box adaptation
as an independent step separated from subtype comparison.

Given the derived type \texttt{tp} and the expected type \texttt{pt},
before we invoke the subtype comparison \texttt{isSubtype(tp, pt)}
to check that the derived type conforms to the expectation,
we first call box adaptation \texttt{adapt(tp, pt)} to transform the type
into \texttt{tp1} with box adaptation.
During this, we traverse the type recursively
and heal the mismatch of boxes as much as possible.
Finally, we invoke the subtype comparison to check whether the adapted type
\texttt{tp1} could conform to the expectation.
For example, given the actual type \mintinline{scala}{{} (op: {io} Unit -> Unit) -> Unit}
and the expected type \mintinline{scala}{{} (op: box {io} Unit -> Unit) -> Unit},
box adaptation transforms the actual type into \mintinline{scala}{{io} (op: box {io} Unit -> Unit) -> Unit}
(which simulates an eta-expansion and an unbox on the term level).
Now we invoke the subtype comparison on the adapted type and the expectation,
which then reports a type-mismatch error.

\subsubsection{Tracking Additional Captured References}

Now we propose an approach to implement box adaptation
by tracking \emph{additional} captured references during the adaptation,
which fits in well with the architecture of capture checking.
The basis of the proposed approach is the following observation.

\emph{Observation.}
Box adaptation only makes the adapted term capture \emph{more} references
(i.e. $\thisx \in C$ if $\adptcDft{T}{U}{\cCat}{C}$)
unless the derivation tree of box adaptation looks like the following.
\begin{equation*}
  \inferrule*[right={box}]
  {\inferrule*[right={box}]
    {\inferrule*[right={box},vdots=1.5em]
      {\inferrule*{\cdots}{\adptcDft{\tCap{C_1}{S_1}}{\tCap{C_2}{S_2}}{\cVal/\cVar}{C}}}
      {\adptcDft{\tCap{C_1}{S_1}}{\Hsquare\,\tCap{C_2}{S_2}}{\cTrm}{\set{}}}}
    {\adptcDft{\tCap{C_1}{S_1}}{\Hsquare\,\cdots\,\Hsquare\tCap{C_2}{S_2}}{\cTrm}{\set{}}}}
  {\adptcDft{\tCap{C_1}{S_1}}{\Hsquare\,\Hsquare\,\cdots\,\Hsquare\tCap{C_2}{S_2}}{\cTrm}{\set{}}}
\end{equation*}
The above derivation first adapts the input to a variable or a value,
and then ends by inserting a chain of boxes.
In such cases, the resulted captured set is always empty.

This observation shows that
to compute the captured variable of the adapted term,
it suffices to record the additional captured references introduced by box adaptation,
and at the same time detect the special cases where boxes are inserted on the top of values,
in which case we should treat the term as if it captures nothing.
In the next we illustrate how we integrate this approach into the existing framework of capture checking.

\begin{listing}[htbp]
	\centering
  \begin{minted}[linenos]{scala}
def recheck(tree: Tree, pt: Type)(using Context): Type =
  val tpe = recheck1(tree, pt)
  checkConforms(tpe, pt)
  tpe
  \end{minted}
  \caption{Pseudo-code of Scala's capture checking implementation.}
  \label{lst:scala-cc-impl}
\end{listing}

Listing \ref{lst:scala-cc-impl} shows the pseudo-code for Scala's capture checker.%
\footnote{\url{https://github.com/lampepfl/dotty/blob/e751f51d8ba8fa33c0a35c26b66ce59466f2bc16/compiler/src/dotty/tools/dotc/cc/CheckCaptures.scala}}
\texttt{recheck1} implements the actual checking logic,
and computes the captured variable of \texttt{tree} with the capturing environments.
The \texttt{checkConforms} function checks whether the actual type \texttt{tpe} conforms to the expectation,
during which box adaptation is called.

To implement box inference, we firstly detect the special case in \texttt{recheck1}
when the expected type is boxed, and the \texttt{tree} is a value.
In that special case we stop the references in \texttt{tree} from being pushed into the capturing environments.
Secondly, during box adaptation, we record the additional capture references in the environments.
For instance, if we insert an unbox $\tUnbox{C}{x}$, the set of references $C$ gets pushed into the environment.




\section{Related Work}

\subsection{Capture Calculus}
The systems presented in this paper are based on the capture calculus \cite{Odersky2022ScopedCF},
which was first introduced in a report on implementing safer exceptions in Scala \cite{Odersky2021}.
The calculus proposes to track the capability for throwing an exception
thereby ensuring that these capabilties will not escape their intended scope.
The report contains an informal discussion of algorithmic typing and box inference,
which outlines the general idea of the box inference implementation in Scala 3.
Subsequently, \cite{Odersky2022ScopedCF} elaborates the capture calculus
and presents the full metatheory.
\cite{BoruchGruszecki2021TrackingCV} presents a similar calculus
that tracks captures in types.
By contrast, it does not have boxes, and type variables could potentially be impure and require being tracked.

\subsection{Algorithmic Typing}
Owing to their close correlation to compiler construction,
the design and decidability of (semi-)algorithmic type systems
have been the subject of extensive research.
The syntax-directed type systems have been proposed and studied for System F
and its variant, F$_{<:},$ which incorporates subtyping \cite{tapl,fsub1,fsub2,hindley1969,milner1978}.
System F$_{<:}$ presents challenges in algorithmic typing
due to the abiguity caused by the transitivity and the subsumption rules.
To circumvent this,
these rules are inlined into other rules to create a syntax-directed ruleset.
As \cc{} is based on System F$_{<:}$,
our system \ccAlgo{} uses the same strategy to achieve syntax-directness,
but it also modifies the rules to facilitate capturing types and employs algorithmic avoidance in \algorn{let}.
However,
it should be noted that the algorithmic typing of System F$_{<:}$ is shown to be undecidable \cite{tapl,fsub1,fsub2}.
Furthermore,
\cite{gadt1,gadt2} investigate the algorithmic type system with the presence of inheritances and variances,
which essentially formalises typechecking algorithms in object-oriented scenarios.
There are other research endeavors that attempt to achieve the algorithmics of typing
in a variety of type systems \cite{dot1,dot2,toopl}.


\section{Conclusion}

In this paper, we develop a family of semi-algorithmic box inference system,
which inserts box operations to heal the mismatch of boxes in the input program.
We start with the semi-algorithmic variant of capture calculus,
named System \ccAlgo{},
which makes the typing and subtyping rules in capture calculus syntax-directed.
Based on it, we develop the box inference system \ccAdp{}
that transforms the terms by inserting the missing box-related constructs
so that the result term is well-typed in capture calculus.
Taking one step further, we propose \ccAdpt{} and show that
it is possible to perform box inference on type level,
operating only on types and not computing the transformed terms,
but with the equivalent expressive power as term-level system \ccAdp{}.
In metatheory we establish the relationship between these systems,
showing that \ccAlgo{} is equivalent to \cc{},
that \ccAdp{} has strictly more expressive power than \ccAlgo{},
and that \ccAdp{} and \ccAdpt{} are equivalent.

\begin{acks}
  We thank
  Aleksander Boruch-Gruszecki, Jonathan Immanuel Brachthäuser, Edward Lee and Ondřej Lhoták
  for their insightful suggestions and constructive feedbacks during the development of this work.
\end{acks}

\bibliography{references.bib}

\appendix



\section{Proof}

\subsection{Well-Formed Environment}



Before carrying out the proof, we first define the notion of well-formedenss.
A capture set is well-formed iff it only includes the program variables in the environment and the root capability $*$,
i.e. ($\wfTyp{\G}{C}$ iff $C \subseteq \dom{\G} \cup \set{*}$).
Type well-formedness is standardly defined,
but additionally requiring that the capture sets are well-formed.
A well-formed environment has all bindings well-formed.

In the following proof, we implicitly assume that all the environments are well-formed,
and it is thus imperative that the transformations on environments should preserve well-formedness.

\subsection{Proof of System \ccAlgo{}}
\label{sec:cc-algo-proof}

\begin{lemma}[Transitivity of subcapturing]
  \label{lemma:transitivity-of-subcapturing}
	For any environment $\G$, if $\subDft{C_1}{C_2}$ and $\subDft{C_2}{C_3}$ then $\subDft{C_1}{C_3}$.
\end{lemma}

This lemma has been proven in \cite{Odersky2022ScopedCF} (Lemma A.12).

\begin{lemma}[Term binding narrowing for subcapturing]
  \label{lemma:term-narrowing-subcapt}
  $\forall \G$, if (1) $\subAlgoDft{T_1}{T_2}$ and (2) $\subAlgo{\G_1, x: T_2, \G_2}{C_1}{C_2}$ then $\subAlgo{\G_1, x: T_1, \G_2}{C_1}{C_2}$.
\end{lemma}

\begin{proof}
  Proceeds by induction on the subcapturing derivation.
  The \ruleref{sc-elem} case can be concluded directly from the precondition and the same rule.
  The \ruleref{sc-set} can be proven by applying the IH repeatedly and invoking the \ruleref{sc-set} rule afterwards.
  For the \textsc{sc-var} case,
  we have $C_1 = \set{y}$ for some $y$,
  $y : U \in \G$
  and $\sub{\G}{\cs{U}}{C_2}$.
  We proceed with a case analysis on whether $x = y$.
  Firstly, if $x \neq y$, the binding is not affected by narrowing,
  which allows us to conclude directly.
  Otherwise, we have $x = y$,
  implying that $\subDft{\cs{T_2}}{C_2}$.
  By inspecting the subtype derivation $\subAlgoDft{T_1}{T_2}$
  we can show that $\subDft{\cs{T_1}}{\cs{T_2}}$.
  Now we invoke Lemma \ref{lemma:transitivity-of-subcapturing} to show that
  $\subDft{\cs{T_1}}{C_2}$,
  and thus conclude the case by applying the \ruleref{sc-var} rule.
\end{proof}

\begin{lemma}[Type binding narrowing for subcapturing]
  \label{lemma:type-narrowing-subcapt}
  $\forall \G$, if (1) $\subAlgoDft{S_1}{S_2}$ and (2) $\sub{\G_1, X <: S_2, \G_2}{C_1}{C_2}$ then $\sub{\G_1, X <: S_1, \G_2}{C_1}{C_2}$.
\end{lemma}

\begin{proof}
  Proven by induction on the subcapturing derivation.
  No rules depend on the type bindings.
\end{proof}

\begin{lemma}[Permutation]  \label{lemma:permutation}
  Given $\G$ and $\Delta$ where $\Delta$ is obtained by permuting $\G$ while preserving well-formedness:
  (1) $\subAlgoDft U T$ implies $\subAlgo{\Delta} U T$,
  (2) $\typAlgDft U T$ implies $\typAlg{\Delta} U T$.
\end{lemma}

\begin{proof}
  By straightforward induction on the derivations,
  wherein all cases but \ruleref{var} and \ruleref{tvar} can be concluded by the IH and the same rule.
  The \ruleref{var} and \ruleref{tvar} cases,
  on the other hand,
  follow from the fact that the ordering of bindings does not affect the result of context lookup.
\end{proof}

\begin{lemma}[Term binding narrowing for algorithmic subtyping]
  \label{lemma:term-narrowing-sub-algo}
  $\forall \G_1, \G_2$, if (1) $\subAlgo{\G_1}{T_1}{T_2}$ and (2) $\subAlgo{\G_1, x: T_2, \G_2}{T}{U}$ then $\subAlgo{\G_1, x: T_1, \G_2}{T}{U}$.
\end{lemma}

\begin{proof}
  By induction on the subtyping derivation $\subAlgo{\G_1, x: T_2, \G_2}{T}{U}$.
  All the cases but the \ruleref{alg-capt} case can be proven immediately with IH and the same rule.
  The \ruleref{alg-capt} case can be proven by applying the IH and Lemma \ref{lemma:term-narrowing-subcapt}.
\end{proof}

%


%

%

\begin{lemma}[Weakening for subcapturing] \label{lemma:weakening-subcapt}
  Given $\G$ and $\Delta$ such that $\G, \Delta$ is well-formed,
  if $\subDft{C_1}{C_2}$, then $\sub{\G, \Delta}{C_1}{C_2}$ given that $\G, \Delta$ is still wellformed.
\end{lemma}

\begin{proof}
  We can prove by induction on the subcapturing derivation.
  Both \textsc{sc-elem} and \textsc{sc-set} can be proven by IH and the same rule.
  For the \textsc{sc-var}, we can conclude by the wellformedness, IH and the same rule.
\end{proof}

\begin{lemma}[Weakening for algorithmic subtyping] \label{lemma:weakening-sub-algo}
  Given $\G$ and $\Delta$ such that $\G, \Delta$ is well-formed,
  if $\subAlgoDft UT$ then $\subAlgo{\G, \Delta}UT$ given that $\G, \Delta$ is still wellformed.
\end{lemma}

\begin{proof}
  By induction on the subtyping derivation.
  \ruleref{alg-refl} and \ruleref{alg-top} can be proven by applying the same rule.
  \ruleref{alg-tvar}, \ruleref{alg-capt} and \ruleref{alg-boxed} can be proven by IH and Lemma \ref{lemma:weakening-subcapt}.

  \emph{Case \ruleref{alg-fun}}.
  In this case $U = \tForall{z}{U_1}{U_2}$ and $T = \tForall{z}{T_1}{T_2}$.
  By IH we can show that $\subAlgo{\G, \Delta}{U_2}{U_1}$
  and $\subAlgo{\G, x: T_1, \Delta}{U_2}{T_2}$.
  By Lemma \ref{lemma:permutation} we can show that $\subAlgo{\G, \Delta, x: T_1}{U_2}{T_2}$.
  We can conclude by \algorn{app}.

  \emph{Case \ruleref{alg-tfun}}.
  In this case, $U = \tTForall{X}{S_1}{T_1}$ and $T = \tTForall{X}{S_2}{T_2}$.
  By IH we have $\subAlgo{\G, \Delta}{S_2}{S_1}$ and
  $\subAlgo{\G, X <: S_2, \Delta}{T_1}{T_2}$.
  By Lemma \ref{lemma:permutation} we have $\subAlgo{\G, \Delta, x <: S_2}{T_1}{T_2}$ and conclude by \algorn{tfun} rule.
\end{proof}

\begin{definition}[Subtype bridge]
  $T$ is a \tnew{subtype bridge} iff $\forall U_1, U_2.$ $\subAlgoDft{U_1}{T}$ and $\subAlgoDft{T}{U_2}$ implies
  $\subAlgoDft{U_1}{U_2}$.
\end{definition}

\begin{lemma}[Weak type binding narrowing] \label{lemma:weak-type-narrowing-sub-algo}
  If (1) $\subAlgoDft{S_1}{S_2}$;
  (2) $\subAlgo{\G, X <: S_2}UT$;
  and (3) $S_2$ is a subtype bridge,
  then
  $\subAlgo{\G, X <: S_1}UT$.
\end{lemma}

\begin{proof}
  Proceeds by induction on the derivation of $\subAlgo{\G, X <: S_2}{U}{T}$,
  wherein all cases but \ruleref{alg-capt} and \ruleref{alg-tvar}
  can be directly concluded by applying the IH and the same rule.
  Now we present the proof of the \ruleref{alg-capt} and the \ruleref{alg-tvar} cases.

  \emph{Case \ruleref{alg-capt}}.
  Then $U = \tCap{C_1}{R_1}$,
  $T = \tCap{C_2}{R_2}$,
  $\subAlgo{\G, X <: S_2}{C_1}{C_2}$
  and $\subAlgo{\G, X <: S_2}{R_1}{R_2}$.
  By Lemma \ref{lemma:type-narrowing-subcapt}
  we demonstrate that $\subAlgo{\G, X <: S_1}{C_1}{C_2}$.
  By the IH we show that
  $\subAlgo{\G, X <: S_1}{R_1}{R_2}$.
  Now we can conclude by the \ruleref{alg-capt} rule.

  \emph{Case \ruleref{alg-tvar}}.
  Then we have $U = Y$ for some type variable $Y$.
  If $Y \neq X$, the binding of $Y$ is not affected by the narrowing and we can conclude by IH.
  If $X = Y$, we have $\subAlgo{\G, X <: S_2}{S_2}{T}$.
  By the IH we have $\subAlgo{\G, X <: S_1}{S_2}{T}$.
  By Lemma \ref{lemma:weakening-sub-algo}, we can show that $\subAlgo{\G, X <: S_1}{S_1}{S_2}$.
  Since $S_2$ is a subtype bridge,
  we can show that $\subAlgo{\G, X <: S_1}{S_1}{T}$ and conclude by the \ruleref{alg-tvar} rule.
\end{proof}

\begin{lemma} \label{lemma:sub-algo-X-Y-S}
  If
  (1) $\subAlgoDft{X}{Y}$; and
  (2) $\subAlgoDft{Y}{S}$
  then $\subAlgoDft{X}{S}$.
\end{lemma}

\begin{proof}
  Proof proceeds by induction on the derivation of the first premise.
  The \ruleref{alg-refl} case is trivial.
  For the \ruleref{alg-tvar} case, we have
  $X <: S_X \in \G$ and $\subAlgoDft{S_X}{Y}$.
  By inverting the judegement $\subAlgoDft{S_X}{Y}$ we have $S_X = Z$ for some type variable $Z$.
  Now we invoke IH again and prove that $\subAlgoDft{Z}{S}$ and conclude by the \ruleref{alg-tvar} rule.
\end{proof}

\begin{lemma} \label{lemma:sub-algo-T-fun-fun}
  If
  (1) $\subAlgoDft{T_1}{\tForall{z}{U_2}{T_2}}$;
  (2) $\subAlgoDft{U_3}{U_2}$;
  (3) $\subAlgo{\G, x: U_3}{T_2}{T_3}$;
  and (4) $U_2$ and $T_2$ are subtype bridges;
  then $\subAlgoDft{T_1}{\tForall{z}{U_3}{T_3}}$.
\end{lemma}

\begin{proof}
  Proof proceeds by induction on the first premise.
  The \ruleref{alg-refl} case is trivial.
  The \ruleref{alg-fun} case can be concluded the premise, Lemma \ref{lemma:term-narrowing-sub-algo} and \ruleref{alg-fun} rule.
  The \ruleref{alg-tvar} case can be concluded by IH.
\end{proof}

\begin{lemma} \label{lemma:sub-algo-T-tfun-tfun}
  If
  (1) $\subAlgoDft{T_1}{\tTForall{X}{S_2}{U_2}}$;
  (2) $\subAlgoDft{S_3}{S_2}$;
  (3) $\subAlgo{\G, X <: S_3}{U_1}{U_2}$;
  and (4) both $S_2$ and $U_2$ are subtype bridges,
  then $\subAlgoDft{T_1}{\tForall{X}{S_3}{U_3}}$.
\end{lemma}

\begin{proof}
  Proceed by induction on the first derivation.
  The \ruleref{alg-refl} is trivial.

  The \ruleref{alg-tfun} case can be concluded by the premise, Lemma \ref{lemma:weak-type-narrowing-sub-algo} and \ruleref{alg-tfun} rule.

  The \ruleref{alg-tvar} case can be concluded by applying the IH.
\end{proof}

\subalgorefl*

\begin{proof}
  Proceed by induction on the structure of $T$.

  \emph{When $T$ is a type variable or $\top$}.
  Both cases can be concluded directly from the \ruleref{alg-refl} and \ruleref{alg-top} rule.

  \emph{When $T = \tCap{C}{S}$}.
  This case can be concluded by IH and reflexivity of subcapturing.

  \emph{Other cases}.
  By the IH and the corresponding subtyping rule.
\end{proof}

\subalgotrans*

\begin{proof}
    By induction on $T_2$.
    \begin{itemize}
    \item \textsc{When $T_2$ is a type variable $X$}.
      Proof proceeds by induction on the judgement $\subAlgoDft{T_1}X$,
      wherein only
      \ruleref{alg-refl} and \ruleref{alg-tvar}
      are applicable.
        \begin{itemize}
        \item \emph{Case \ruleref{alg-refl}}.
          Then $T_1 = X$ too.
          The goal is trivially concluded by the premise.
        \item \emph{Case \ruleref{alg-tvar}}.
          We conclude immediately by Lemma \ref{lemma:sub-algo-X-Y-S}.
        \end{itemize}

      \item \textsc{When $T_2$ is $\top$}.
        By inspecting $\subAlgoDft{\top}{T_3}$ we can show that $T_3 = \top$,
        and thus conclude by the \ruleref{alg-top} rule.

      \item \textsc{When $T_2$ is a function $\tForall{x}{U_1}{U_2}$}.
        Proof proceeds by case analysis the second premise.
        The only possible cases are \ruleref{alg-top}, \ruleref{alg-refl} and \ruleref{alg-fun}.
        In the case of \ruleref{alg-top}, we conclude by the same rule.
        In the \ruleref{alg-refl} rule, the goal follows from the premise.
        For the \ruleref{alg-fun} case, we invoke Lemma \ref{lemma:sub-algo-T-fun-fun} to conclude.
        \item \textsc{When $T_2$ is a type function $\tTForall{X}{S}{U}$}.
        Proceeds by case analysis on the second premise.
        The \ruleref{alg-top} and \ruleref{alg-refl} cases are trivial.
        The \ruleref{alg-tfun} case can be concluded by Lemma \ref{lemma:sub-algo-T-tfun-tfun}.
        Other cases are not applicable when $T_2$ is a type function.

        \item \textsc{When $T_2$ is a boxed type $\tBox{U}$}.
        We proceed by induction on the premise $\subAlgoDft{T_1}{\tBox{U}}$.
        The three possible cases are \ruleref{alg-refl}, \ruleref{alg-tvar} and \ruleref{alg-boxed}.
        The \ruleref{alg-refl} case is trivial.
        In the \ruleref{alg-tvar} case,
        we can invoke the IH and conclude by the same rule.
        In the \ruleref{alg-boxed} case, we have $T_1 = \tBox{U_1}$ for some $U_1$, and $\subAlgoDft{U_1}{U}$.
        We invert the premise $\subAlgoDft{\tBox{U}}{T_3}$ and obtain three cases: \textsc{Refl}, \textsc{Top} and \textsc{Boxed}.
        Both the \ruleref{alg-refl} and the \ruleref{alg-top} cases are trivial.
        In the \ruleref{alg-boxed} case we get $T_3 = \tBox{U_3}$ for some $U_3$ and $\subAlgoDft{U}{U_3}$.
        We can show that $\subAlgoDft{U_1}{U_3}$ using the IH.
        We can prove the goal by applying the \ruleref{alg-boxed} rule then.

        \item \textsc{When $T_2$ is a capturing type $\tCap{C}{S}$}.
        We proceed the proof by induction on the premise $\subAlgoDft{T_1}{\tCap{C}{S}}$.
        There are three cases: \ruleref{alg-refl}, \ruleref{alg-tvar} and \ruleref{alg-capt}.
        The \ruleref{alg-refl} case is trivial.
        \ruleref{alg-tvar} case can be proven with the IH and the \ruleref{alg-tvar} rule.
        In the \ruleref{alg-capt} rule we have $T_1 = \tCap{C_1}{S_1}$ for some $C_1, S_1$, and $\subDft{C_1}{C}, \subAlgoDft{S_1}{S}$.
        We invert the premise $\subAlgoDft{\tCap C S}{T_3}$ and have three cases: \ruleref{alg-tvar}, \ruleref{alg-top} and \ruleref{alg-capt}.
        Both \ruleref{alg-refl} and \ruleref{alg-top} are trivial.
        For the \ruleref{alg-capt} case we have $T_3 = \tCap{C_3}{S_3}$ $\subDft{C}{C_3}$, and $\subAlgoDft{S}{S_3}$.
        With the IH we can prove that $\subAlgoDft{S_1}{S_3}$ and since the subcapture is transitive we also have $\subDft{C_1}{C_3}$.
        Now we apply the \ruleref{alg-capt} rule to prove the goal.
    \end{itemize}
\end{proof}


\begin{theorem}[Completeness of algorithmic subtyping $\subAlgoDft{T_1}{T_2}$]
  \label{thm:completeness-of-algorithmic-subtyping}
  $\forall T_1, T_2, \G$, if $\subDft{T_1}{T_2}$ then $\subAlgoDft{T_1}{T_2}$.
\end{theorem}

\begin{proof}
  The proof proceeds by induction on premise $\subDft{T_1}{T_2}$
  \begin{itemize}
    \item Case \ruleref{refl}.
      By Lemma \ref{lemma:sub-algo-refl}.
    \item Case \ruleref{trans}.
      Then we have some $U$ such that $\subAlgoDft{T_1}{U}$ and $\subAlgoDft{U}{T_2}$
      by IH.
      We can conclude this case using Lemma \ref{lemma:algorithmic-subtyping-transitivity}.
    \item Case \ruleref{tvar}.
      We have $\subDft{X}{T_2}$ in this case, and $X <: T_2 \in \G$.
      We can prove the goal with \ruleref{alg-tvar} and \ruleref{alg-refl}.
    \item Case \ruleref{top}. Trivial.
    \item Case \ruleref{fun}, \ruleref{tfun}, \ruleref{capt} and \ruleref{boxed}.
      These cases can be concluded by apply the IH and the corresponding rule.
  \end{itemize}
\end{proof}

\begin{theorem}[Soundness of algorithmic subtyping $\subAlgoDft{T_1}{T_2}$]
  \label{thm:soundness-of-algorithmic-subtyping}
  For any $\G, T_1, T_2$, if $\subAlgoDft{T_1}{T_2}$ then $\subDft{T_1}{T_2}$.
\end{theorem}

\begin{proof}
    By induction on $\subAlgoDft{T_1}{T_2}$.
    All but one case can be proven by the IH
    and the corresponding rule.
    The only interesting case is \ruleref{alg-tvar}.
    We have $T_1 = X, \subAlgoDft{S}{T_2}$ and $X <: S \in \G$ for some $X, S$.
    By inductive hypothesis we know that $\subDft{S}{T_2}$.
    So we can conclude this case with \ruleref{tvar} and \ruleref{trans}.
\end{proof}



\begin{lemma}[Term substitution perserves subcapturing]
  \label{lemma:term-subst-subcapturing}
  If $\sub{\G_1, x: U, \G_2}{C_1}{C_2}$ and $\subDft{D}{\cv{U}}$,
  then $\sub{\G_1, [x := D] \G_2}{[x := D]C_1}{[x := D]C_2}$.
\end{lemma}

\begin{lemma}[Type substitution perserves subcapturing]
  \label{lemma:type-subst-subcapturing}
  If $\sub{\G_1, X <: R, \G_2}{C_1}{C_2}$ and $\sub{\G_1}{P}{R}$,
  then $\sub{\G_1, [X := P] \G_2}{C_1}{C_2}$.
\end{lemma}


Both lemmas have been established in the metatheory of \cite{Odersky2022ScopedCF}.

\begin{lemma}[Term subtitution perserves subtyping]
  \label{lemma:term-subst-subtyping}
  If $\subAlgo{(\G_1, z : U, \G_2)}{T_1}{T_2}$, $\typAlgDft{y}{U^\prime}$ and $\subAlgo{\G_1}{U^\prime}{U}$, then $\subAlgo{(\G_1, [z := y]\G_2)}{[z := y] {T_1}}{[z := y]} T_2$.
\end{lemma}

\newcommand{\theenv}{\G_1, z: U, \G_2}

\begin{proof}
	We proceed the proof through induction on
  $\subAlgo{(\G_1, z : U, \G_2)}{T_1}{T_2}$.
  \begin{itemize}
    \item Case \ruleref{alg-refl} and \ruleref{alg-top}.
    By the same rule.

    \item Case \ruleref{alg-capt}.
      In this case, $T_1 = \tCap{C_1}{S_1}$ and $T_2 = \tCap{C_2}{S_2}$.
      And we have $\sub{\theenv}{C_1}{C_2}$ and $\subAlgo{\theenv}{S_1}{S_2}$.
      From $\subAlgo{\G_1}{U^\prime}{U}$ we know that $\sub{\G_1}{\cs{U^\prime}}{\cs{U}}$.
      Therefore $\subAlgo{\G_1}{\set{y}}{\cs{U}}$.
      We can therefore invoke Lemma \ref{lemma:term-subst-subcapturing}
      to show that
      $\sub{(\G_1, [z := y]\G_2)}{[z := y] \cs{U^\prime}}{[z := y] \cs{U}}$,
      and thus conclude the case by the IH and the \ruleref{alg-capt} rule.

    \item Case \ruleref{alg-boxed}, \ruleref{alg-fun} and \ruleref{alg-tfun}.
      These cases can be proven by the IH and the same rule.

    \item Case \ruleref{alg-tvar}.
    In this case $T_1 = X$, $X <: S \in \theenv$ and $\subAlgo{\theenv}{S}{T_2}$.
    By IH we know that $\sub{\G_1, [z := y] \G_2}{[z := y] S}{[z := y] T_2}$.
    If $X <: S \in \G_1$, then by the well-formness of $\theenv$ we know that $z$ does not appear in $S$.
    Therefore, $[z := y] S = S$ and we can conclude using the \ruleref{tvar} rule.
    If $X <: S \in \G_2$, then we know that $X <: [z := y] \in [z := y] \G_2$,
    which allows us to conclude by the \ruleref{tvar} rule as well.
  \end{itemize}
\end{proof}

\begin{lemma}[Type subtitution perserves subtyping]
  \label{lemma:type-subst-subtyping}
  If $\subAlgo{(\G_1, X <: S, \G_2)}{T_1}{T_2}$, then $\subAlgo{(\G_1, [X := S]\G_2)}{[X := S] {T_1}}{[X := S]} T_2$.
\end{lemma}

\renewcommand{\theenv}{\G_1, X <: S, \G_2}

\begin{proof}
	Proof proceeds by induction on the subtyping derivation.
  \begin{itemize}
  \item \emph{Case \ruleref{alg-refl}, \ruleref{alg-top}}.
    By the same rule.

  \item \emph{Case \ruleref{alg-capt}}.
    Then $T_1 = \tCap{C_1}{S_1}$ and $T_2 = \tCap{C_2}{S_2}$.
    This case can be concluded by Lemma \ref{lemma:type-subst-subcapturing}, the IH and the same rule.

  \item \emph{Case \ruleref{alg-boxed}, \ruleref{alg-fun} and \ruleref{alg-tfun}}.
    Prove by the IH and the re-application of the same rule.

  \item \emph{Case \ruleref{alg-tvar}}.
    In this case $T_1 = Y$, $Y <: S_1 \in \theenv$ and $\subAlgo{\theenv}{S_1}{T_2}$.
    Proceed with a case analysis on where $Y$ is bound.
    \begin{itemize}
    \item \emph{Case $X = Y$}.
      In this case $S_1 = S$.
      The goal becomes $\subAlgo{\G_1, [X := S] \G_2}{S}{\fsubst X S T_2}$.
      By the IH, we can show that
      $\subAlgo{\G, \fsubst X S \G_2}{\fsubst X S S}{\fsubst X S T_2}$.
      By the wellformness of the environment we know that $S \in \G_1$ and $X$ does not appear in $S$.
      Therefore $[X := S] S = S$, thus concluding this case.
    \item \emph{Case $Y <: S_1 \in \G_1$}.
      In this case, by the well-formness we also know that $X$ does not appear in $S_1$,
      thus $[X := S] S_1 = S_1$.
      We can conclude by the IH and the \ruleref{alg-tvar} rule.

    \item \emph{Case $Y <: S_1 \in \G_2$}.
      In this case, $Y <: [X := S] S_1 \in [X := S] \G_2$.
      We can conclude with the IH and the \ruleref{alg-tvar} rule.
    \end{itemize}
  \end{itemize}
\end{proof}

\begin{lemma}[Subtype inversion: capturing type]
  \label{lemma:capturing-inv}
  If $\subAlgoDft{T_1}{\tCap{C_2}{S_2}}$,
  then $T_1 = \tCap{C_1}{S_1}$, $\subDft{C_1}{C_2}$ and $\subAlgoDft{S_1}{S_2}$.
\end{lemma}

\begin{proof}
  Proof proceeds by induction on $\subAlgoDft{T_1}{\tCap{C_2}{S_2}}$,
  wherein only the \ruleref{alg-refl} and the \ruleref{alg-capt} cases
  are possible,
  and in both cases we conclude trivially.
\end{proof}




%

\begin{lemma}[Subtyping inversion: functions]
  \label{lemma:inv-subtype-forall}
	If $\subAlgoDft{S}{\tForall{x}{U}{T}}$, then
  either $S = X$, $\typUpDft{X}{\tForall{x}{U^\prime}{T^\prime}}$ and $\subAlgoDft{\tForall{x}{U^\prime}{T^\prime}}{\tForall{x}{U}{T}}$;
  or $S = \tForall{x}{U^\prime}{T^\prime}$ and $\subAlgoDft{\tForall{x}{U^\prime}{T^\prime}}{\tForall{x}{U}{T}}$.
\end{lemma}

\begin{proof}
	We proceed the proof by induction on $\subAlgoDft{S}{\tForall x U T}$.
  There are three possible cases: \ruleref{alg-refl}, \ruleref{alg-tvar} and \ruleref{alg-fun}.
  \begin{itemize}
    \item \ruleref{alg-refl}. Trivial.
    \item \ruleref{alg-tvar}.
    In this case, $S = X$, $X <: S_1 \in \G$ and $\subAlgoDft{S_1}{\tForall x U T}$.
    By invoking IH we know that either
    either $S_1 = X_1$, $\typUpDft{X_1}{\tForall{x}{U^\prime}{T^\prime}}$ and $\subAlgoDft{\tForall{x}{U^\prime}{T^\prime}}{\tForall{x}{U}{T}}$;
    or $S = \tForall{x}{U^\prime}{T^\prime}$ and $\subAlgoDft{\tForall{x}{U^\prime}{T^\prime}}{\tForall{x}{U}{T}}$.
    In the first case, we can show that $\typUpDft{X}{\tForall{x}{U^\prime}{T^\prime}}$.
    So the goal is proven.
    We can also conclude in the second case.

    \item \ruleref{alg-fun}.
    In this case
    $S = \tForall{x}{U^\prime}{T^\prime}$ and $\subAlgoDft{\tForall{x}{U^\prime}{T^\prime}}{\tForall{x}{U}{T}}$.
    So we conclude directly.
  \end{itemize}
\end{proof}


\begin{lemma}[Subtyping inversion: polymorphic functions]
  \label{lemma:inv-subtype-poly}
  If $\subAlgoDft{T_1}{\tTForall{X}{S_2}{U_2}}$, then
  either $T_1 = X$, $\typUpDft{X}{\tTForall{X}{S_1}{U_1}}$, and $\subAlgoDft{\tTForall{X}{S_1}{U_1}}{\tTForall{X}{S_2}{U_2}}$;
  or $T_1 = \tTForall{X}{S_1}{U_1}$ and $\subAlgoDft{\tTForall{X}{S_1}{U_1}}{\tTForall{X}{S_2}{U_2}}$.
\end{lemma}

\begin{proof}
	Proof proceeds with induction on the subtyping judgement.
  \begin{itemize}
    \item \emph{Case \ruleref{alg-refl} and \ruleref{alg-fun}.}
      This two cases can be concluded trivially.

    \item \emph{Case \ruleref{alg-tvar}.}
    In this case $T_1 = X$, $X <: S \in \G$ and $\subAlgoDft{S}{\tTForall{X}{S_2}{U_2}}$.
    By IH we either
    know $S = Y$, $\typUpDft{Y}{\tTForall{X}{S_1}{U_1}}$, and $\subAlgoDft{\tTForall{X}{S_1}{U_1}}{\tTForall{X}{S_2}{U_2}}$;
    or $S = \tTForall{X}{S_1}{U_1}$ and $\subAlgoDft{\tTForall{X}{S_1}{U_1}}{\tTForall{X}{S_2}{U_2}}$.
    In both cases we can derive $\typUpDft{X}{\tTForall{X}{S_1}{U_1}}$
    and also $\subAlgoDft{\tTForall{X}{S_1}{U_1}}{\tTForall{X}{S_2}{U_2}}$.
  \end{itemize}
\end{proof}

\begin{lemma}[Subtyping inversion: Boxed types]
	\label{lemma:inv-subtype-box}
  If $\subAlgoDft{T_1}{\tBox{U_2}}$, then
  either $T_1 = X$, $\typUpDft{X}{\tBox{U_1}}$ and $\subAlgoDft{\tBox{U_1}}{\tBox{U_2}}$;
  or $T_1 = \tBox{U_1}$ and $\subAlgoDft{\tBox{U_1}}{\tBox{U_2}}$.
\end{lemma}

\begin{proof}
	Proof proceeds by induction on the subtyping derivation.
  \begin{itemize}
  \item \emph{Case \ruleref{alg-refl} and \ruleref{alg-boxed}}. Trivial.

  \item \emph{Case \ruleref{alg-tvar}}.
    In this case, $T_1 = X$, $X <: S \in \G$ and $\subAlgoDft{S}{\tBox{U_2}}$.
    By IH we know that
    either $S_1 = Y$, $\typUpDft{Y}{\tBox{U_1}}$ and $\subAlgoDft{\tBox{U_1}}{\tBox{U_2}}$;
    or $S_1 = \tBox{U_1}$ and $\subAlgoDft{\tBox{U_1}}{\tBox{U_2}}$.
    In both cases $\typUpDft{X}{\tBox{U_1}}$ and $\subAlgoDft{\tBox{U_1}}{\tBox{U_2}}$.
  \end{itemize}
\end{proof}

\begin{lemma} \label{lemma:tvar-upcast-precise}
  $\forall \G$, if $\typUpDft{X}{S}$, then $\forall T$ that is not a type variable, $\subAlgoDft{X}{T}$ implies $\subAlgoDft{S}{T}$.
\end{lemma}

\begin{proof}
  Proof proceeds by induction on the upcasting relationship.

  \emph{Case \ruleref{widen-shape}}.
  We then invert the subtype judgement and only find two possible cases: \ruleref{alg-top} and \ruleref{alg-tvar}.
  The \ruleref{alg-top} case is trivial.
  In the \ruleref{alg-tvar} case we have $\subAlgoDft{S}{T}$ and conclude immediately.

  \emph{Case \ruleref{widen-tvar}}.
  We similarly invert the subtype judgement and find two cases.
  The \ruleref{alg-top} case is again trivial.
  For the \ruleref{alg-tvar} case we have $\subAlgoDft{Y}{T}$ and invoke IH to conclude.
\end{proof}

\begin{lemma} \label{lemma:sub-rhs-tvar-inv}
  If $\subAlgoDft{S}{X}$ then $S$ is a type variable.
\end{lemma}

\begin{proof}
  We prove by inverting the derivation.

  \emph{Case \ruleref{alg-top}}. Conclude immediately.

  \emph{Case \ruleref{alg-tvar}}.
  In this case we have $S = Y$ for some $Y$ and conclude.
\end{proof}

\begin{lemma} \label{lemma:sub-tvar-upcast-to-same}
  If $\subAlgoDft{X}{Y}$ and $\typUpDft{Y}{S}$ then $\typUpDft{X}{S}$.
\end{lemma}

\begin{proof}
  Proof by induction on the subtype judgement.

  \emph{Case \ruleref{alg-refl}}. We conclude immediately.

  \emph{Case \ruleref{alg-tvar}}.
  In this case we have $X <: R \in \G$ and $\subAlgoDft{R}{Y}$.
  By Lemma \ref{lemma:sub-rhs-tvar-inv} we know that $R = Z$ for some $Z$.
  By IH we have $\typUpDft{Z}{S}$ and conclude by applying \ruleref{widen-tvar}.

  Other cases are not applicable here.
\end{proof}

\begin{lemma}[Term binding narrowing for widen variable typing] \label{lemma:term-narrowing-typ-up}
  $\forall \G$, if (1) $\subAlgoDft{T_1}{T_2}$ and (2) $\typAlg{\G, x: T_2, \Delta}{t}{U}$, then $\typUp{\G, x: T_1, \Delta}{y}{U^\prime}$ and $\subAlgoDft{U^\prime}{U}$.
\end{lemma}

\begin{proof}
  If the narrowed and typed variable is not the same one, then the typing is not affected and we conclude immediately.
  Otherwise, if the variable is narrowed, we do case analysis on what $T_2$ is.
  If $T_2$ is not a type variable, we know that $U = T_2$.
  If $T_1$ is not a type variable, we conclude immediately.
  If $T_1$ is a type variable we conclude with Lemma \ref{lemma:tvar-upcast-precise}.
  Otherwise if $T_2$ is a type variable, 
  we analyze the type of $T_1$.
  If $T_1$ is not a type variable again we conclude immediately.
  Otherwise, we conclude by Lemma \ref{lemma:sub-tvar-upcast-to-same}.
\end{proof}

\begin{lemma} \label{lemma:widen-to-tvar-absurd}
  If $\typUpDft{X}{S}$ then $S$ is not a type variable.
\end{lemma}

\begin{proof}
  Proof proceeds by induction on the derivation. 
  For the \ruleref{widen-shape} case we conclude immediately.
  For the \ruleref{widen-tvar} we conclude by IH.
\end{proof}

\begin{lemma} \label{lemma:subcapt-indom}
  If $C \subseteq \dom{\G}$ and $\subDft{C^\prime}{C}$ then $C^\prime \subseteq \dom{\G}$.
\end{lemma}

\begin{proof}
  Proof proceeds by induction on the derivation of the subcapturing judgement.

  \emph{Case \ruleref{sc-elem}}.
  We conclude immediately by $\set{x} \subseteq C \subseteq \dom{\G}$.

  \emph{Case \ruleref{sc-set}}.
  Conclude by IH.

  \emph{Case \ruleref{sc-var}}.
  Conclude immediately from the premise.
\end{proof}




\begin{proof}
  By induction on the derivations,
  wherein all cases can be concluded by IH and the same rule.
  Notably,
  the \ruleref{var} and \ruleref{tvar} case can be concluded due to the fact that
  the mentioned variables are not dropped from the environment.
\end{proof}

\begin{lemma}[Avoidance implies subtyping]
  \label{lemma:avoidance-sub}
  If $\G(x) = \tCap{C}{S}$ and $U = \avoidOp{x}{C}{T}$ then $\subDft{T}{U}$.
\end{lemma}

\begin{proof}
  We prove by structural induction on type $T$.
  All cases can be proven by IH and the corresponding rule.
\end{proof}

\begin{lemma} \label{lemma:narrowing-avoidance-subcapt}
  If $\subAlgo{\G, x: \tCap{C_1}{S}, \Delta}{D_1}{D_2}$ and $\subDft{C_1}{C_2}$,
  then for $D_i^\prime = \avoidOp{x}{C_i}{D_i}$ we have $\subAlgo{\G, x: \tCap{C_1}{S}, \Delta}{D_1^\prime}{D_2^\prime}$.
\end{lemma}

\begin{proof}
	By straightforward induction on the derivation.

  \emph{Case \ruleref{sc-elem}}.
  Then $D_1 = \set{y}$ and $y \in D_2$.
  If $x = y$ then
  the goal becomes $\subAlgo{\G, x : \tCap{C_2}{S}, \Delta}{C_1}{\fsubst{x}{C_2}D_2}$
  at covariant occurrences
  and is $\subAlgo{\G, x : \tCap{C_2}{S}, \Delta}{\set}{D_2 \setminus \set{x}}$
  at contravariant occurrences.
  The contravariant case is trivial,
  whereas the covariant case can be concluded by using the premise the weakening lemma.

  \emph{Case \ruleref{sc-set}}. By repeated IH.

  \emph{Case \ruleref{sc-var}}. By the IH.
\end{proof}

\begin{lemma} \label{lemma:narrowing-avoidance}
  If $\subAlgo{\G, x: \tCap{C_1}{S}, \Delta}{U_1}{U_2}$ and $\subDft{C_1}{C_2}$,
  then for $U_i^\prime = \avoidOp{x}{C_i}{U_i}$ we have $\subAlgo{\G, x: \tCap{C_1}{S}, \Delta}{U_1^\prime}{U_2^\prime}$.
\end{lemma}

\begin{proof}
  By straightforward induction on the subtype derivation,
  wherein all cases but \ruleref{alg-capt} can be concluded by IH and the same rule.
  For the \ruleref{alg-capt} case, we use Lemma \ref{lemma:narrowing-avoidance-subcapt} to conclude.
\end{proof}

\begin{lemma}[Avoidance]
	Given any environment $\G, x : T$,
  and any $U''$ such that $x \notin \fv{U''}$ and $\sub{\G, x : T}{U}{U''}$,
  we have $\sub{\G}{U'}{U''}$.
  Here $U' = \avoidOp{x}{\cv{T}}{U}$.
\end{lemma}

\begin{proof}
  By induction on the subtype derivation,
  wherein all cases can be concluded by the IH and the same rule.
\end{proof}

\newcommand{\sandwichEnv}[1]{\G, x: {#1}, \Delta}

\begin{lemma}[Term binding narrowing for algorithmic typing]
	\label{lemma:narrowing-typ-algo}
  $\forall \G$, if (1) $\subAlgoDft{T_1}{T_2}$ and (2) $\typAlg{\G, x: T_2, \Delta}{t}{U}$, then $\typAlg{\G, x: T_1, \Delta}{t}{U^\prime}$ and $\subAlgo{\sandwichEnv{T_1}}{U^\prime}{U}$.
\end{lemma}

\begin{proof}
  Proceeds by induction on the derivation of the typing judgement.

  \emph{Case \ruleref{alg-var}}.
  If the binding that gets narrowed is the same variable in the typing judgement,
  we conclude immediately from the premise.
  Otherwise, we also conclude immediately since the binding is not affected.

  \emph{Case \ruleref{alg-abs} and \ruleref{alg-tabs}}.
  Conclude by the IH, the same rule, \ruleref{alg-fun} and \ruleref{alg-tfun}.

  \emph{Case \ruleref{alg-app}}.
  By Lemma \ref{lemma:term-narrowing-typ-up} we have $\typUp{\G, x: T_1, \Delta}{x}{U_0}$ and 
  $\subAlgo{\sandwichEnv{T_1}}{U_0}{C\,\tForall{z}{U_1}{U_2}}$.
  By Lemma \ref{lemma:capturing-inv} and Lemma \ref{lemma:inv-subtype-forall} we have
  $U_0 = C^\prime\,\tForall{z}{U_1^\prime}{U_2^\prime}$,
  $\sub{\sandwichEnv{T_1}}{C^\prime}{C}$, $\subAlgo{\sandwichEnv{T_1}}{U_1}{U_1^\prime}$ and $\subAlgo{\sandwichEnv{T_1}, z: U_1}{U_2^\prime}{U_2}$.
  By IH we have $\typAlg{\sandwichEnv{T_1}}{t}{T_0}$ and $\subAlgo{\sandwichEnv{T_1}}{T_0}{T^\prime}$.
  By Lemma \ref{lemma:algorithmic-subtyping-transitivity} we have $\subAlgo{\sandwichEnv{T_1}}{T_0}{U_1^\prime}$.
  Therefore, we can type the result as $\tSubst{z}{y}{U_2^\prime}$.
  We can show that $\subAlgo{\sandwichEnv{T_1}}{\tSubst{z}{y}{U_2^\prime}}{\tSubst zy{U_2}}$
  by Lemma \ref{lemma:term-subst-subtyping}.

  \emph{Case \ruleref{alg-tapp}}. Similarly to \ruleref{alg-app}.

  \emph{Case \ruleref{alg-box}}.
  Conclude by IH, Lemma \ref{lemma:capturing-inv}, Lemma \ref{lemma:subcapt-indom} and the same rule.

  \emph{Case \ruleref{alg-unbox}}.
  By Lemma \ref{lemma:term-narrowing-typ-up}, Lemma \ref{lemma:capturing-inv} and Lemma \ref{lemma:inv-subtype-box} we can show that $\typUp{\sandwichEnv{T_1}}{x}{C_x^\prime\,\tBox{\tCap{C_0}{S_0}}}$,
  $\sub{\sandwichEnv{T_1}}{C_0}{C}$ and $\sub{\sandwichEnv{T_1}}{S_0}{S}$.
  By Lemma \ref{lemma:transitivity-of-subcapturing} we can show that $\sub{\sandwichEnv{T_1}}{C_0}{C}$.
  We can then apply the \ruleref{alg-unbox} rule can type the result as $\tCap{C_0}{S_0}$ and conclude.

  \emph{Case \ruleref{alg-let}}. 
  By IH we can show that $\typAlg{\sandwichEnv{T_1}}{s}{T^\prime}$ and $\subAlgo{\sandwichEnv{T_1}}{T^\prime}{T}$.
  We again invoke IH and show that $\typAlg{\sandwichEnv{T_1}, y: T}{t}{U_1}$
  and $\subAlgo{\sandwichEnv{T_1}, y: T}{U_1}{U}$.
  Let $U_1^\prime = \avoidOp{x}{\cv{T^\prime}}{U_1}$.
  By Lemma \ref{lemma:narrowing-avoidance} we can show that $\subAlgo{\sandwichEnv{T_1}}{U_1^\prime}{U^\prime}$ and conclude.
\end{proof}

\begin{theorem}[Completeness of algorithmic typing] \label{thm:typ-algo-completeness}
  If $\typDft{t}{T}$, then $\typAlgDft{t}{T^\prime}$ for some $T^\prime$ and $\subAlgoDft{T^\prime}{T}$.
\end{theorem}

\begin{proof}
  The proof proceeds by induction on premise $\typDft{t}{T}$.
  \begin{itemize}
    \item Case \ruleref{var}.
    In this case we have $t = x$ for some $x$ and $\G(x) = \tCap C S$.
      Also, $T = \tCap{\set{x}}{S}$.
    We can apply the \ruleref{alg-var} rule to prove the goal.

    \item Case \ruleref{sub}.
    In this case we have $\typDft{t}{U}$ and $\subDft{U}{T}$ for some $U$.
    With the inductive hypothesis we know that $\typAlgDft t {U^\prime}$ for some $U^\prime$ and $\subAlgoDft{U^\prime}U$.
    By applying Theorem \ref{thm:completeness-of-algorithmic-subtyping} on $\subDft U T$ we have $\subAlgoDft U T$.
    By applying Lemma \ref{lemma:algorithmic-subtyping-transitivity} we can prove $\subAlgoDft{U^\prime}T$.

    \item Case \ruleref{abs}.
    In this case, $t = \tLambda{x}{U_1}{u}$, $T = \tForall{x}{U_1}{U_2}$ and $\typ{\extendG{x}{U_1}}{u}{U_2}$.
    By inductive hypothesis we know that $\typAlg{\extendG{x}{U_1}}{u}{U_2^\prime}$ and $\subAlgo{\extendG{x}{U_1}}{U_2^\prime}{U_2}$.
      By applying the rule in the algorithmic typing system we can prove that $\typAlgDft{t}{\cv{t}/x\,\tForall{x}{U_1}{U_2^\prime}}$.
    By applying the \ruleref{alg-fun} we can prove that $\subAlgoDft{\tForall{x}{U_1}{U_2^\prime}}{\tForall{x}{U_1}{U_2}}$.
    Then we conclude by applying \ruleref{alg-capt} and the reflexivity of subcapturing.

    \item Case \ruleref{tabs}.
    Here, $t = \tTLambda X S u$, $T = \tTForall{X}{S}{U}$, and $\typ{\extendGT{X}{S}}uU$.
    By inductive hypothesis we have $\typAlg{\extendGT X S}{u}{U^\prime}$ and $\subAlgo{\extendGT X S}{U^\prime}U$.
    We can show that $\typAlgDft{t}{\cv{t}\,\tTForall X S {U^\prime}}$ and $\subAlgoDft{\tTForall X S {U^\prime}}{\tTForall X S U}$.
      Now we conclude by applying \ruleref{alg-capt} and reflexivity of subcapturing.

    \item Case \ruleref{app}.
    Here, $t = x\,y$.
    By inductive hypothesis we know that $\typAlgDft{x}{U_1}$, $\typAlgDft{y}{U_2}$, $\subAlgoDft{U_1}{C\,\tForall{x}{U}{T}}$ and ${\subAlgoDft{U_2}{U}}$.
    We first apply Lemma \ref{lemma:capturing-inv} to show that $U_1 = \tCap{C_1}{S_1}$ and $\subDft{C_1}{C_2} \wedge \subAlgoDft{S_1}{\tForall x U T}$.
    Then, with Lemma \ref{lemma:inv-subtype-forall} we know that
    either $S_1 = X$, $\typUpDft{X}{\tForall{x}{U^\prime}{T^\prime}}$ and $\subAlgoDft{\tForall{x}{U^\prime}{T^\prime}}{\tForall{x}{U}{T}}$;
    or $S_1 = \tForall{x}{U^\prime}{T^\prime}$ and $\subAlgoDft{\tForall{x}{U^\prime}{T^\prime}}{\tForall{x}{U}{T}}$.
    In both cases, we can show that $\typUpDft{x}{\tForall{x}{U^\prime}{T^\prime}}$,
    and in both cases we have
    $\subAlgoDft{\tForall{x}{U^\prime}{T^\prime}}{\tForall{x}{U}{T}}$.
      By inverting this judgement we have $\subAlgoDft{U}{U^\prime}$ and $\subAlgo{\G, z: U}{T^\prime}{T}$.
    By applying Lemma \ref{lemma:algorithmic-subtyping-transitivity}, we can show that $\subAlgoDft{U_2}{U^\prime}$.
    We use the \ruleref{alg-app} rule to derive $\typAlgDft{x\,y}{[z := y] T^\prime}$.
      We can conclude by invoking Lemma \ref{lemma:term-subst-subtyping} on $\subAlgo{\G, z: U}{T^\prime}{T}$.

    \item Case \ruleref{tapp}.
    In this case $t = x [S]$.
    By IH we know that $\typAlgDft{x}{U_1}$ and $\subAlgoDft{U_1}{C\,\tTForall{X}{S}{T}}$.
    By Lemma \ref{lemma:inv-subtype-poly} we know that
    either $U_1 = X$, $\typUpDft{X}{\tTForall{X}{S^\prime}{T^\prime}}$, and $\subAlgoDft{\tTForall{X}{S^\prime}{T^\prime}}{\tTForall{X}{S}{T}}$;
    or $U_1 = \tTForall{X}{S^\prime}{U^\prime}$ and $\subAlgoDft{\tTForall{X}{S^\prime}{U^\prime}}{\tTForall{X}{S}{U}}$.
    In both cases we have $\typUpDft{x}{\tTForall{X}{S^\prime}{U^\prime}}$.
    And in both cases we can invert the subtype judgement.
    By applying the \ruleref{alg-tapp} we can derive that $\typAlgDft{x[S]}{[X := S] T^\prime}$.
    We conclude by invoking Lemma \ref{lemma:type-subst-subtyping} on $\subAlgo{\G, X <: S}{T^\prime}{T}$.

    \item Case \ruleref{box}.
    $t = \tBox{x}$.
    By IH we know that $\typAlgDft{x}{U_1}$ and $\subAlgoDft{U_1}{{\tCap {C_2} {S_2}}}$.
    By Lemma \ref{lemma:capturing-inv} we know that $U_1 = \tCap {C_1} {S_1}$, $\subDft{C_1}{C_2}$ and $\subAlgoDft{S_1}{S_2}$.
    With \ruleref{alg-box} we can derive that $\typAlgDft{t}{\tBox{\tCap{C_1}{S_1}}}$.
    We conclude with $\subAlgoDft{\tBox{\tCap{C_1}{S_1}}}{\tBox{\tCap{C_2}{S_2}}}$.

    \item Case \ruleref{unbox}.
    $t = \tUnbox{C}{x}$.
    By IH we know that $\typAlgDft{x}{U_1}$ and $\subAlgoDft{U_1}{\tBox{\tCap C {S}}}$.
    Now we invoke Lemma \ref{lemma:inv-subtype-box} and show that $\typUpDft{x}{\tBox{\tCap{C^\prime}{S^\prime}}}$ in both cases, where we have $\subAlgoDft{\tCap{C^\prime}{S^\prime}}{\tCap{C}{S}}$.
    Now we invoke the \ruleref{alg-unbox} rule to derive $\typAlgDft{x}{\tCap{C^\prime}{S^\prime}}$ and conclude.

    \item Case \ruleref{let}.
    $t = \tLet{x}{s}{t}$.
    By IH we know that $\typAlgDft{s}{T^\prime}$ and $\subAlgoDft{T^\prime}{T}$.
    Also $\typAlg{\extendG{x}{T}}{t}{U^\prime}$ and $\subAlgo{\extendG{x}{T}}{U^\prime}{U}$.
    By Lemma \ref{lemma:narrowing-typ-algo} we can show that $\typAlg{\extendG{x}{T^\prime}}{t}{U^\prime}$.
    By the \ruleref{alg-let} we can derive that $\typAlgDft{t}{U_a}$ where $U_a = \avoidOp{x}{\cv{T^\prime}}{U^\prime}$ given that $T^\prime = \tCap{C^\prime}{S^\prime}$.
    So the goal is to show that $\subAlgoDft{U_a}{U}$.
    We can conclude with the avoidance lemma.


  \end{itemize}
\end{proof}



\begin{lemma}
  \label{lemma:widen-to-general}
	If $\typUpDft{X}{S}$ then $\subDft{X}{S}$.
\end{lemma}

\begin{proof}
  Proof proceeds by induction on the derivation of the premise.
  \begin{itemize}
  \item \emph{Case \ruleref{widen-shape}}.
    Conclude with the \ruleref{tvar} rule.
  \item
    \emph{Case \ruleref{widen-tvar}}.
    Conclude with IH and \ruleref{trans}.
  \end{itemize}
\end{proof}

\begin{lemma} \label{lemma:up-to-typ}
	If $\typUpDft{x}{T}$ then $\typDft{x}{T}$.
\end{lemma}

\begin{proof}
  We invert the premise and get two cases.
  In the \ruleref{var-widen} case we conclude with Lemma \ref{lemma:widen-to-general} and \ruleref{sub} rule.
  In the \ruleref{var-lookup} we conclude with \ruleref{var}.
\end{proof}

\begin{theorem}[Soundness of Algorithmic Typing] \label{thm:typ-algo-soundness}
	If $\typAlgDft{t}{T}$ then $\typDft{t}{T}$.
\end{theorem}

\begin{proof}
	Proof proceeds by induction on the derivation of the typing judgement.
  \begin{itemize}
  \item \emph{Case \ruleref{alg-var}}.
    Conclude with Lemma \ref{lemma:up-to-typ}.

  \item
    \emph{Case \ruleref{alg-abs}, \ruleref{alg-tabs}, \ruleref{alg-box} and \ruleref{alg-unbox}}.
    Conclude with IH and the corresponding rule.

  \item
    \emph{Case \ruleref{alg-app}}.
    By Lemma \ref{lemma:widen-to-general} we have
    $\typDft{x}{C\,\tForall{x}{T}{U}}$.
    By Lemma \ref{thm:soundness-of-algorithmic-subtyping} we have
    $\subDft{T^\prime}{T}$.
    By \ruleref{alg-fun} and \ruleref{alg-sub} we have
    $\typDft{x}{C\,\tForall{x}{T^\prime}{U}}$.
    We conclude with \ruleref{app}.

  \item
    \emph{Case \ruleref{alg-tapp}}.
    As above.

  \item
    \emph{Case \ruleref{alg-let}}.
    We conclude by applying \ruleref{let} and Lemma \ref{lemma:avoidance-sub}.

  \end{itemize}
\end{proof}

\subsection{Proof of System \ccAdp{}}
\label{sec:cc-adp-proof}

\subsubsection{Properties of term normalisation}

\begin{lemma} \label{lemma:alg-typ-term-subst}
  If $\typAlg{\G_1, x: T, \G_2}{t}{U}$ and $\typAlg{\G}{y}{T}$ then $\typAlg{\G_1, [x := y]\G_2}{[x:=y]t}{[x:=y]U}$.
\end{lemma}

\begin{proof}
  By Lemma \ref{thm:typ-algo-soundness} and \ref{thm:typ-algo-completeness} we know that the algorithmic system and the original system is equivalent.
  Therefore, we get this Lemma for free from the proof of the original system.
\end{proof}

\begin{lemma} \label{lemma:typ-up-to-prec}
  If $\typUpDft{x}{T}$ then $x : T^\prime \in \G$ and $\subDft{T^\prime}{T}$.
\end{lemma}

\begin{proof}
  Proceed by examining the derivation of $\typUpDft{x}{T}$.
  In the case of \ruleref{var-lookup},
  we can conclude by the reflexivity of subtyping.
  In the \ruleref{var-widen} case,
  we can invoke Lemma \ref{lemma:widen-to-general} to conclude.
\end{proof}

\begin{lemma} \label{lemma:compact-let}
  If $\typAlgDft{t}{T}$ and $t^\prime = \compactLet{t}$,
  then $\exists T^\prime, \typAlgDft{t^\prime}{T^\prime}$ and $\subDft{T^\prime}{T}$.
\end{lemma}

\begin{proof}
	By case analysis on the $\compactLet{t}$ function.

  \emph{Case \eruleref{nl-box}.}
  Then $t = \tLet{y}{\tUnbox{C}{x}}{\tBox{y}}$,
  and $t^\prime = x$.
  By inverting the typing judgement $\typAlgDft{\tLet{y}{\tUnbox C x}{\tBox y}}{T}$ we know that $\typUpDft{x}{\tBox{\tCap {C^\prime} S}}$, $\typAlg{(\G, y : \tCap{C^\prime}{S})}{\tBox{y}}{\tBox{\tCap{C^\prime}{S}}}$, $\subDft{C}{C^\prime}$ and $T = \tBox{\tCap{C^\prime}{S}}$.
  We conclude by Lemma \ref{lemma:typ-up-to-prec}.

  \emph{Case \eruleref{nl-deref}.}
  In this case $t = \tLet{y}{s}{y}$.
  By inverting the typing judgement $\typAlgDft{t}{T}$, we have
  $\typAlgDft{s}{T}$ and $\typAlg{(\G, y: T)}{y}{T}$.
  Note that since $y$ is fresh, $y$ does not appear as a free variable in $T$ so avoidance has no effect.
  We can then conclude directly.

  \emph{Case \eruleref{nl-rename}.}
  In this case $t = \tLet{x}{y}{s}$.
  By inverting the typing judgement we have
  $\typAlgDft{x}{\tCapSet{y}{S}}$, $\typ{(\G, x: \tCapSet{y}{S})}{t}{U}$
  and $T = \avoidOp{x}{\set{y}}{U}$.
  By Lemma \ref{lemma:alg-typ-term-subst} we have
  $\typ{\G}{[x:=y]t}{[x:=y]U}$.
  Since in the covariant place the substitution will be identical and in contravariant place $\subDft{\eset{}}{\set{y}}$ we have $\subDft{[x:=y]U}{T}$.
\end{proof}

\begin{lemma} \label{lemma:compact-lam}
  If $\typAlgDft{t}{T}$ and $t^\prime = \compactLam{t}$,
  then $\exists T^\prime, \typAlgDft{t^\prime}{T^\prime}$ and $\subDft{T^\prime}{T}$.
\end{lemma}

\begin{proof}
	Proof proceeds by case analysis on $\compactLam{t}$.

  \emph{Case \eruleref{nl-beta}.}
  In this case $t = \tLambda{z}{U}{x\,z}$,
  and $t^\prime = x$.
  Then we invert the typing judgement and get
  $\typUp{(\G, z: U)}{x}{\tForall{y}{U^\prime}{T^\prime}}$, $\subAlgoDft{U}{U^\prime}$ and
  $T = \tForall{z}{U}{[y:=z]T^\prime}$.
  With alpha conversion we have $T = \tForall{y}{U}{T^\prime}$.
  We conclude with the \textsc{Fun} rule.
  Now we conclude by Lemma \ref{lemma:typ-up-to-prec}
  and the transitivity of subtyping.

  \emph{Case \eruleref{nl-tbeta}.}
  As above.

\end{proof}

\begin{corollary} \label{lemma:compact}
  If $\typAlgDft{t}{T}$ and $t^\prime = \fEmbed{t}$,
  then $\exists T^\prime, \typAlgDft{t^\prime}{T^\prime}$ and $\subDft{T^\prime}{T}$.
\end{corollary}


\subsubsection{Metatheory of Box Adaptation}

\begin{wide-rules}

\textbf{Algorithmic subtyping with subcapturing inlined \quad $\subCaptDft{T}{U}$}

\begin{multicols}{2}

\infrule[\ruledef{algc-refl}]
  {\subDft{C}{C^\prime}}
  {\subCaptDft{\tCap{C}{S}}{\tCap{C^\prime}{S}}}

\infrule[\ruledef{algc-tvar}]
  {X <: S \in \G \\
   \subCaptDft{\tCap{C}{S}}{\tCap{C^\prime}{S^\prime}}}
  {\subCaptDft{\tCap{C}{X}}{\tCap{C^\prime}{S^\prime}}}

\infrule[\ruledef{algc-top}]
  {\subDft{C}{C^\prime}}
  {\subCaptDft{\tCap{C}{S}}{\tCap{C^\prime}{\top}}}

\infrule[\ruledef{algc-boxed}]
  {
    \subCaptDft{U_1}{U_2} \\
    \subDft{C_1}{C_2}
  }
  {\subCaptDft{C_1\,\tBox{U_1}}{\tCap{C_2}{\tBox{U_2}}}}

\end{multicols}

\begin{multicols}{2}

\infrule[\ruledef{algc-fun}]
  {\subCaptDft{U_2}{U_1} \\
    \subCapt{\extendG{x}{U_2}}{T_1}{T_2} \\
    \subDft{C}{C^\prime}
  }
  {\subCaptDft{C\,\tForall{x}{U_1}{T_1}}{C^\prime\,\tForall{x}{U_2}{T_2}}}

\infrule[\ruledef{algc-tfun}]
  {
    \subCaptDft{S_2}{S_1} \\
    \subCapt{\extendGT{X}{S_2}}{T_1}{T_2} \\
    \subDft{C}{C^\prime}
  }
  {\subCaptDft{C\,\tTForall{X}{S_1}{T_1}}{C^\prime\,\tTForall{X}{S_2}{T_2}}}

\end{multicols}

\caption{Algorithmic subtyping rules with subcapturing inlined}
  \label{fig:capt-sub}

\end{wide-rules}

To ease the development of the metatheory for box adaptation,
we make use of an auxiliary subtyping relation
which is derived from algorithmic subtyping by inling the \textsc{capt} rule.
The rules are presented in Figure \ref{fig:capt-sub}.
We begin by establishing its equivalence with algorithmic subtyping.

\begin{lemma}[Soundness of $\vdash_{C\rightarrow}$]
  \label{lemma:capt-sub-to-algo}
	If $\subCaptDft{T}{U}$ then $\subAlgoDft{T}{U}$.
\end{lemma}

\begin{proof}
	Proof proceeds by induction on the derivation of the subtyping judgement,
  wherein all cases can be concluded by applying the IH and the \textsc{capt} rule.
\end{proof}

\begin{lemma}[Completeness of $\vdash_{C\rightarrow}$]
  \label{lemma:algo-sub-to-capt}
  If $\subAlgoDft{T}{U}$ then $\subCaptDft{T}{U}$.
\end{lemma}

\begin{proof}
	Proof proceeds by induction on the subtype derivation.
  All cases but \textsc{capt} can be concluded directly from the IH and the corresponding rule.
  For the \textsc{capt} case, we can invert the subtyping judgement between shape types, apply the IH and corresponding rule again to conclude.
\end{proof}




\begin{lemma}[Weakening of box adaptation]
  \label{lemma:adp-weakening}
  If $\adpDft{x}{T}{t_x}{U}$,
  then $\adp{\G, \Delta}{x}{T}{t_x}{U}$.
\end{lemma}

\begin{proof}
  By straightforward induction on the derivation.
  In each case we can conclude by the IH and the same rule.
\end{proof}

\begin{lemma}
  \label{lemma:box-adaptation-completeness-capt-sub}
  If
  $\subCaptDft{T}{U}$ then
  $\adpDft{x}{T}{x}{U}$.
\end{lemma}

\begin{proof}
	Proof proceeds by induction on the derivation of the subtype judgement.

  \emph{Case \ruleref{algc-refl} and \ruleref{algc-top}.}
  These two cases follows directly from the corresponding rule in \ccAdp{}.

  \emph{Case \ruleref{algc-tvar}}.
  Conclude by using the IH and the corresponding rule in \ccAdp{}.

  \emph{Case \ruleref{algc-boxed}}.
  Then $T = \tCap{C_1}{\tBox{U_1}}$,
  $U = \tCap{C_1}{\tBox{U_2}}$,
  $\subDft{C_1}{C_2}$,
  and $\subCaptDft{U_1}{U_2}$.
  Let $t = \tLet{y}{\tUnbox{C_1}{x}}{\tLet{z}{y}{\tBox{z}}}$.
  By applying the IH,
  we can show that $\adpDft{y}{U_1}{y}{U_2}$.
  By the definition of $\fEmbed{t}$,
  we have $\fEmbed{t} = x$.
  We conclude by applying the \ruleref{ba-boxed} rule.

  \emph{Case \ruleref{algc-fun}}.
  Then $T = \tCap{C_1}{\forall(z : U_1)T_1}$,
  $U = \tCap{C_2}{\forall(z : U_2)T_2}$,
  $\subCaptDft{U_2}{U_1}$,
  $\subCapt{(\G, z : U_2)}{T_1}{T_2}$,
  and $\subDft{C_1}{C_2}$.
  Let $t = \lambda(z : U_2) \tLet{z^\prime}{z}{\tLet{y}{x\, z^\prime}{z}}$.
  By IH we have
  $\adp{\G}{z}{U_2}{z}{U_1}$
  and $\adp{(\G, z: U_2)}{y}{T_1}{y}{T_2}$.
  By Lemma \ref{lemma:adp-weakening},
  we can show that
  $\adp{\G, z: U_2, z': U_1}{y}{T_1}{y}{T_2}$.
  By the definition of $\fEmbed{t}$,
  we can show that
  $\fEmbed{t} = x$.
  Therefore $\fsubst x {C_1} \cv{\fEmbed{t}} = C_1$.
  We can conclude using the \ruleref{ba-fun} rule.

  \emph{Case \ruleref{algc-tfun}}.
  As above.
\end{proof}

\subadpcompleteness*

\begin{proof}
  Conclude by Lemma \ref{lemma:algo-sub-to-capt} and \ref{lemma:box-adaptation-completeness-capt-sub}.
\end{proof}

\newcommand{\exenv}{\G, \Delta}

\begin{lemma}  \label{lemma:adp-sub-soundness-1}
  If $\adpDft{x}{T}{t_x}{U}$ and $T, U$ are well-formed,
  then $\forall{\Delta}$, if $\typAlg{\G, \Delta}{x}{T}$ then $\typ{\exenv}{t_x}{U}$.
\end{lemma}

\begin{proof}
  Proof by induction on the derivation of the adaptation subtyping.

  \emph{Case \ruleref{ba-refl} and \ruleref{ba-top}}.
  Conclude by applying the \ruleref{var} and the \ruleref{sub} rules.

  \emph{Case \ruleref{ba-tvar}}.
  Then $T = \tCap{C}{S}$,
  $U = \tCap{C^\prime}{S^\prime}$,
  $X <: \tCap{C}{S} \in \G$,
  and $\adpDft{x}{\tCap{C}{S}}{t_x}{\tCap{C^\prime}{S^\prime}}$.
  By invoking the IH, we can show that
  $\forall \Delta$ such that $\typAlg{\G, \Delta}{x}{\tCap{C}{S}}$,
  we have $\typ{\G, \Delta}{t_x}{\tCap{C^\prime}{S^\prime}}$.
  Now, we consider any $\Delta$ such that $\typAlg{\G}{x}{\tCap{C}{X}}$,
  implying that $\Delta = \Delta_1, x : \tCap{C_0}{X}, \Delta_2$
  and $C = \set{x}$.
  We first set $\Delta^\prime = \Delta_1, x : \tCap{C_0}{S}, \Delta_2$,
  and use IH to show that $\typ{\G, \Delta'}{t_x}{\tCap{C'}{S'}}$.
  Now we apply narrowing to show that $\typ{\G, \Delta}{t_x}{\tCap{C'}{S'}}$
  and thus conclude this case.

  \emph{Case \ruleref{ba-boxed}}.
  Then
  $T = \tCap{D}{\tBox{\tCap{C}{S}}}$,
  $U = \tCap{D'}{\tBox{\tCap{C'}{S'}}}$,
  $\adpDft{y}{\tCap{C}{S}}{t_y}{\tCap{C'}{S'}}$,
  and $t_x = \tLet{y}{\tUnbox{C}{x}}{\tLet{z}{t_y}{\tBox z}}$.
  Consider $\Delta$ such that $\typAlg{\G, \Delta}{x}{\tCap{D'}{\tBox{\tCap{C}{S}}}}$.
  We now invoke the IH to show that
  $\typ{\exenv, y: \tCap{C}{S}}{t_y}{\tCap{C'}{S'}}$.
  Therefore, we can derive that
  $\typ{\exenv}{t}{\tCap{D'}{\tCap{C'}{S'}}}$.
  Now we can conclude this case using Lemma \ref{lemma:compact} and the \ruleref{sub} rule.

  \emph{Case \ruleref{ba-fun}}.
  Then $T = \tCap{C}{\forall(z : U_1) T_1}$,
  $U = \tCap{C'}{\forall(z : U_2) T_2}$,
  $\adpDft{z}{U_2}{t_z}{U_1}$,
  $\adp{\G, z : U_1, z' : U_2}{y}{T_1'}{t_y}{T_2}$
  where $T'_1$ is defined as in the rule,
  $t_x = \lambda(z : U_2) \tLet{z'}{t_z}{\tLet{y}{x\,z'}{t_y}}$,
  and $\sub{\G}{\fsubst x C \cv{\fEmbed{t_x}}}{C'}$.
  Consider $\Delta$ such that $x : \tCap{C}{\forall(z : U_1) T_1} \in \Delta$.
  By invoking the IH we can first show that
  $\typ{\exenv, z: U_2}{t_z}{U_1}$.
  If $t_z$ is a variable,
  we can show that
  $\typ{\exenv, z : U_2, z' : U_1}{x\,z}{T_1}$.
  Otherwise, we can show that
  $\typ{\exenv, z : U_2, z' : U_1}{x\,z'}{\fsubst{z}{z'} T_1}$.
  Now we invoke the IH again to show that
  $\typ{\exenv, z : U_2, z' : U_1, y : T'_1}{t_y}{T_2}$
  in both cases.
  We can then derive that $\typ{\exenv}{t_x}{\cv{t_x}\,\tForall{x}{U_2}{T_2}}$.
  Since $\typAlg{\exenv}{x_f}{C\,\tForall{x}{U_1}{T_1}}$, we can show that $C = \set{x_f}$.
  Therefore $\substIn{x_f}{\set{x_f}}{\cv{t_f}} = \cv{t_f}$.
  If $t_z$ is a variable,
  we let $t'_x = \lambda(z: U_2) \tLet{y}{x\,z}{t_y}$,
  and we can show that $\fEmbed{t'_x} = \fEmbed{t_x}$.
  We can show that $\typ{\exenv}{t_x}{C'\,\tForall{x}{U_2}{T_2}}$.
  Otherwise, we can now derive that $\typ{\exenv}{t_x}{C^\prime\,\tForall{x}{U_2}{T_2}}$.
  We conclude by Lemma \ref{lemma:compact} and \textsc{Sub}.

  \emph{Case \ruleref{ba-tfun}}.
  This case can be proven in the same way as the \ruleref{ba-fun} case.

  \emph{Case \ruleref{ba-box}}.
  Then $T = \tCap{C}{S}$,
  $U = \tCap{D}{\tBox{\tCap{C'}{S'}}}$,
  $\adpDft{x}{\tCap{C}{S}}{s_x}{\tCap{C'}{S'}}$,
  and $t_x = \tLet{y}{s_x}{\tBox{y}}$.
  By invoking the IH we can show that
  $\typ{\exenv, x: \tCap CS}{s_x}{\tCap{C^\prime}{S^\prime}}$.
  By the well-formedness, we can show that $y \notin \fv{\tCap{C'}{S'}}$,
  and therefore we can conclude this case by the \ruleref{let} and the \ruleref{sub} rule.

  \emph{Case \ruleref{ba-unbox}}.
  Then $T = \Box\ \tCap{C}{S}$,
  $U = \tCap{C'}{S'}$,
  $\adpDft{y}{\tCap{C}{S}}{t_y}{\tCap{C'}{S'}}$,
  and $t_x = \tLet{y}{\tUnbox{C}{x}}{t_y}$.
  By the IH we can show that
  $\typ{\exenv, y: \tCap{C}{S}}{t_y}{\tCap{C^\prime}{S^\prime}}$.
  Again by the wellformedness we can show that
  $y \notin \tCap{C'}{S'}$,
  and thus conclude by \ruleref{let}.
\end{proof}

\subadpsoundness*

\begin{proof}
  By Lemma \ref{lemma:adp-sub-soundness-1}.
\end{proof}

\subsubsection{Metatheory of Typing with Box Inference}

\begin{lemma}  \label{lemma:ub-to-typ}
  If $\typUbDft{x}{t_x}{T}$ then $\typDft{t_x}{U}$.
\end{lemma}

\begin{proof}
  Proof proceeds by case analysis on the premise.
  The \ruleref{var-return} case follows immediately from Lemma \ref{lemma:up-to-typ}.
  In the \ruleref{var-unbox} case we can invoke Lemma \ref{lemma:up-to-typ} and the \ruleref{unbox} rule to conclude.
\end{proof}

\begin{lemma} \label{lemma:adp-sub-var}
  If $\adpDft{x}{U}{y}{T}$ then $x = y$.
\end{lemma}

\begin{proof}
  Proof by induction on the derivation of the judgement.

  \emph{Case \ruleref{ba-refl} and \ruleref{ba-top}}.
  Immediate.

  \emph{Case \ruleref{ba-tvar}}.
  By IH.

  \emph{Case \ruleref{ba-boxed}}.
  Then $U = \tCap{D}{\tBox{\tCap{C}{S}}}$,
  $T = \tCap{D'}{\tBox{\tCap{C'}{S'}}}$,
  $\adpDft{y'}{\tCap{C}{S}}{t_y}{\tCap{C'}{S'}}$,
  and $y = \fEmbed{t}$
  where $t = \tLet{y'}{\tUnbox{C}{x}}{\tLet{z}{t_y}{\tBox{z}}}$.
  Through the case analysis on $t$ and the definition of $\fEmbed{\cdot}$,
  we can shwo that $t_y = y''$ for some $y''$.
  By IH we can show that $y' = y$ and thus conclude.

  \emph{Case \ruleref{ba-fun}}.
  Then $U = \tCap{C_1}{\forall(z : U_1} T_1$,
  and $T = \tCap{C_2}{\forall(z: U_2) T_2}$,
  $\adp{\G}{z}{U_2}{t_z}{U_1}$,
  $\adp{\G, z : U_2, z' : U_1}{y'}{\fsubst z {z'} T_1}{t_y}{T_2}$,
  and $y = \fEmbed{t_f}$
  where $t_f = \lambda(z : U_2) \tLet{z'}{t_z}{\tLet{y'}{x\,z'}{t_y}}$.
  Through a case analysis on $t_f$ we can show that both $t_z$ and $t_y$ are variables,
  otherwise $\fEmbed{t_f}$ cannot be a variable.
  Now we invoke IH to show that $t_z = z$ and $t_y = y'$,
  which allow us to conclude.

  \emph{Case \ruleref{ba-tfun}}.
  As above.

  \emph{Case \ruleref{bi-box} and \ruleref{bi-unbox}}.
  We can see that $t$ being a variable is impossible.
\end{proof}

\begin{theorem}[Soundness]  \label{thm:typ-adp-soundness}
	If $\typAlgoDft{t}{t^\prime}{T}$, then $\typDft{t^\prime}{T}$.
\end{theorem}

\begin{proof}
	Proof proceeds by induction on the derivation of the premise.

  \emph{Case \ruleref{bi-var}}.
  Then $t = t' = x$ for some $x$.
  We conclude immediately using the \ruleref{var} rule.

  \emph{Case \ruleref{bi-abs} and \ruleref{bi-tabs}}.
  By IH and the corresponding rules.

  \emph{Case \ruleref{bi-app}}.
  Then $t = x\,y$,
  $\typUbDft{x}{t_x}{C\,\forall(z: U) T'}$,
  $\adpTypDft{y}{t_y}{U}$,
  and $t' = \fEmbed{\tLet{x'}{t_x}{\tLet{y'}{t_y}{x'\,y'}}}$.
  By Lemma \ref{lemma:ub-to-typ} we have $\typDft{t_x}{C\,\tForall{z}{T}{U}}$.
  By Lemma \ref{thm:adp-sub-soundness} we have $\typDft{t_y}{U}$.
  Then we perform a case analysis on $t_y$.
  \begin{itemize}
    \item \emph{When $t_y$ is a variable}.
      By Lemma \ref{lemma:adp-sub-var} we know that $t_y = y$.
      In this case $T^\prime = [z:=y]T$.
      Based on the definition of $\fEmbed{\cdot}$,
      we can show that $t' = \fEmbed{\tLet{x'}{t_x}{x'\,y}}$.
      Note that $x \notin \fv{T}$.
      We conclude by the \ruleref{let} rule.

    \item \emph{When $t_y$ is not a variable}.
      In this case $T^\prime = \avoidOp{y^\prime}{\cv{U}}{[z:=y^\prime]U}$.
      We conclude by the \ruleref{let} rule.
  \end{itemize}

  \emph{Case \ruleref{bi-tapp}}.
  Then $t = x[S']$,
  $\typUbDft{x}{t_x}{C\,\forall[X <: S] T'}$ where $T = \fsubst X {S'} {T'}$,
  and  $t' = \tLet{x'}{t_x}{x'[S']}$.
  By Lemma \ref{lemma:ub-to-typ} we have $\typDft{t_x}{C\,\tTForall{X}{S}{T}}$.
  Since $x^\prime$ is fresh, ${x^\prime} \notin \fv{S^\prime} \cup \fv{T}$.
  We can thus derive that $\typDft{t^\prime}{[X := S^\prime] T}$
  by \ruleref{let} and \ruleref{tapp}.

  \emph{Case \ruleref{bi-box}}.
  Then $t = t' = \tBox{x}$,
  $\typUpDft{x}{\tCap{C}{S}}$
  and $C \subseteq \dom{\G}$.
  By Lemma \ref{lemma:up-to-typ} we can show that
  $\typDft{x}{\tCap{C}{S}}$.
  Now we can conclude this case by the \ruleref{box}.

  \emph{Case \ruleref{bi-unbox}}.
  Then $t = t' = \tUnbox{C'}{x}$,
  $\typUpDft{x}{\tBox{\tCap{C}{S}}}$,
  $C \subseteq \dom{\G}$,
  and $\subDft{C}{C'}$.
  By Lemma \ref{lemma:up-to-typ} we can show that
  $\typDft{x}{\tBox{\tCap{C}{S}}}$.
  By \ruleref{sub}, \ruleref{boxed} and \ruleref{capt} we can show that
  $\typ{x}{\tBox{\tCap{C'}{S}}}$.
  This case can be concluded by the \ruleref{unbox} rule.

\end{proof}

\begin{theorem}[Completeness]  \label{thm:typ-adp-completeness}
  If $\typAlgDft{t}{T}$, then $\typAlgoDft{t}{t^\prime}{T}$.
\end{theorem}

\begin{proof}
  Proof proceeds by induction on the derivation of the premise.

  \emph{Case \ruleref{var}}.
  Conclude by applying the \ruleref{bi-var} rule.

  \emph{Case \ruleref{abs}, \ruleref{tabs}, \ruleref{box}, \ruleref{unbox} and \ruleref{let}}.
  We conclude by IH and the corresponding rule.

  \emph{Case \ruleref{app}}.
  We can show that $\typUbDft{x}{x}{C\,\tForall z T U}$ by \textsc{var-return} rule.
  By Lemma \ref{lemma:sub-algo-to-adp} we have $\adpDft{y}{T^\prime}{y}{T}$,
  which leads us to $\adpTypDft{y}{y}{T}$.
  Let $t = \tLet{x^\prime}{x}{\tLet{y^\prime}{y}{x^\prime\,y^\prime}}$.
  We can show that $t^\prime = x\,y$ and $[z := y]T$.
  Therefore we can derive $\typAlgoDft{x\,y}{x\,y}{[z := y]T}$.

  \emph{Case \ruleref{tapp}}.
  As above.
\end{proof}

\subsection{Proof of System \ccAdpt{}}
\label{sec:cc-adpt-proof}

\subsubsection{Properties of Type-Level Box Adaptation}

\begin{lemma}  \label{lemma:drop-thisx}
  $C[x] / \set{x} = C / \set{\thisx, x}$.
\end{lemma}

\begin{proof}
	Direct from the definition.
\end{proof}

\begin{lemma} \label{lemma:sub-adpt-soundness-inner}
  If $\adptcDft{T}{U}{\cCat}{C}$,
  then for any variable $x$, exists $t_x$, such that
  (1) $\adpDft{x}{T}{U}{t_x}$,
  (2) $\fCat{t_x} = \cCat$,
  and (3) $\cv{t_x} = C[x]$.
\end{lemma}

\begin{proof}
  By induction on the derivation.

  \emph{Case \ruleref{t-ba-refl} and \ruleref{t-ba-top}}.
  Then $\cCat = \cVar$ and $C = \set{\thisx}$ in these cases.
  Since $\fCat{x} = \cVar$ and $\set{\thisx}[x] = \set{x}$,
  we can conclude immediately by applying the corresponding rules (\ruleref{ba-refl} and \ruleref{ba-top}).

  \emph{Case \ruleref{t-ba-tvar}}.
  By the IH.

  \emph{Case \ruleref{t-ba-boxed}}.
  Then $T = D_1\,\Box\,\tCap{C_1}{S_1}$,
  $U = D_2\,\Box\,\tCap{C_2}{S_2}$,
  and $\adptcDft{\tCap{C_1}{S_1}}{\tCap{C_2}{S_2}}{\cCat_0}{C_0}$.
  By the IH, we can show that $\exists t_y$ such that $\adpDft{y}{\tCap{C_1}{S_1}}{\tCap{C_2}{S_2}}{t_y}$,
  $\fCat{t_y} = \cCat$ and $\cv{t_y} = C[y]$.
  We define $t = \tLet{y}{\tUnbox{C}{x}}{\tLet{z}{t_y}{\tBox{z}}}$, and
  $t^\prime = \compact{t}$.
  We proceed by case analysis on $\cCat$.
  If $\cCat = \cVar$, in which case $t_y = z$ for some $z$.
  Then by Lemma \ref{lemma:adp-sub-var} we have $z = y$.
  In this case $t^\prime = x$.
  Since in this case $\cCat^\prime = X$, and $C^\prime = \set{\thisx}$,
  we conclude by the \ruleref{ba-boxed} rule.
  In other two cases of $\cCat$, we conclude by the definition of the normalization function and the \ruleref{ba-boxed} rule.

  \emph{Case \ruleref{t-ba-fun}}.
  Then $T = C_1\,\forall(x : U_1) T_1$,
  $U = C_2\,\forall(x: U_2) T_2$,
  $\adptcDft{U_2}{U_1}{\cCat_1}{C_1}$,
  and $\adptc{\G, x : U_2, x' : U_1}{T'_1}{T_2}{\cCat_2}{C_2}$.
  $T'_1$ are defined as in the rule,
  and $\cCat$ and $C$ is computed from $\cCat_1, \cCat_2$ and $C_1, C_2$.
  Assume that the input variable is $x_f$.
  By IH we have $\adpDft{x}{U_2}{t_x}{U_1}$ where $\fCat{t_x} = \cCat_1$, and $\cv{t_x} = C_1[x]$.
  Additionally, we have $\adp{\G, x: U_2, x' : U_1}{z}{T'_1}{t_z}{T_2}$ where $\fCat{t_z} = \cCat_2$ ,and $\cv{t_z} = C_2[z]$.
  Now let's define $t_f = \tLambda{x}{U_1}{\tLet{x^\prime}{t_x}{\tLet{z}{x_f\,x^\prime}{t_z}}}$.
  Also, we define $t_f^\prime = \fEmbed{t_f}$.
  By case analysis we can show that $\cv{t_f^\prime} = \cv{t_f} = \cv{t_x} \cup \cv{t_z} \cup \cv{x_f} / \set{x, x^\prime, z}$.
  Therefore, $[x_f \mapsto C] \cv{t_f^\prime} = \cv{t_x} \cup \cv{t_z} \cup C / \set{x, x^\prime, z, x_f}$.
  By Lemma \ref{lemma:drop-thisx},
  we can show that $C_f[x_f] = \cv{t'_f}$,
  where $C_f = C_1 \cup C_2 \setminus \set{x, x'} \cup \set{\thisx}$
  as defined in the rule.
  By case analysis on $\cCat_1$ and $\cCat_2$ we can show that $\cCat^\prime = \cv{t_f^\prime}$.
  Finally, we conclude this case by applying the \ruleref{ba-fun} rule.

  \emph{Case \ruleref{t-ba-tfun}}. As above.

  \emph{Case \ruleref{t-ba-box}}.
  Then $T = C_1\,C_2$,
  $U = D\,\Box\,C_2\,S_2$,
  $\adptcDft{C_1\,S_1}{C_2\,S_2}{\cCat_0}{C_0}$,
  and $\cCat = \cTrm$.
  By IH we have $\adpDft{x}{\tCap{C_1}{S_1}}{t_x}{\tCap{C_2}{S_2}}$,
  $\cCat = \fCat{t_x}$, and $C_0[x] = \cv{t_x}$.
  We conclude by case analysis on $\cCat$
  and using the \ruleref{ba-box} rule.

  \emph{Case \ruleref{t-ba-unbox}}.
  As above.
\end{proof}

\begin{theorem}[Soundness]  \label{thm:sub-adpt-soundness}
  If $\adptTypDft{x}{T}{\cCat}{C}$, then exists $t_x$, such that
  (1) $\adpTypDft{x}{t_x}{T}$,
  (2) $\fCat{t_x} = \cCat$,
  and (3) $C = \cv{t_x}$.
\end{theorem}

\begin{proof}
  By Lemma \ref{lemma:sub-adpt-soundness-inner}.
\end{proof}

\begin{lemma}  \label{lemma:sub-adapt-completeness-inner}
	If $\adpDft{x}{T}{t_x}{U}$, then
  $\adptcDft{T}{U}{\fCat{t_x}}{C[x]}$ for some $C$.
\end{lemma}

\begin{proof}
	By induction on the derivation of $\adpDft{x}{T}{t_x}{U}$.

  \emph{Case \ruleref{ba-refl} and \ruleref{ba-top}}.
  Conclude immediately by the \ruleref{t-ba-refl} or the \ruleref{t-ba-top} rules.

  \emph{Case \ruleref{ba-tvar}}.
  Conclude by IH, and the \ruleref{t-ba-tvar} rule.

  \emph{Case \ruleref{ba-boxed}}.
  By IH we have $\adptcDft{\tCap{C_1}{S_1}}{\tCap{C_2}{S_2}}{\cCat}{C}$,
  $\fCat{t_y} = \cCat$,
  and $\cv{t_y} = C[y]$.
  If $\fCat{t_y} = \cVar$,
  by Lemma \ref{lemma:adp-sub-var} we can show that $t_y = y$.
  Therefore $C = \{\thisx\}$, and we can conclude.
  The other two cases follows directly from the definitions.

  \emph{Case \ruleref{ba-fun}}.
  By IH we have $\adptcDft{U_2}{U_1}{\cCat_1}{C_1}$,
  and $\adptc{\G, x: U_2, x': U_1}{T_1'}{T_2}{\cCat_2}{C_2}$,
  where $T'_1$ is defined in the rule.
  Proceed by case analysis on $\cCat_1$ and $\cCat_2$.
  If both $\cCat_1$ and $\cCat_2$ are variables,
  we invoke Lemma \ref{lemma:adp-sub-var} and show that $t_x = x$, and $t_z = z$.
  Therefore $t_f^\prime = x_f$,
  and $\cv{t_f^\prime} = \set{x_f}$.
  Also, we can show that $C_1 = C_2 = \set{\thisx}$.
  Therefore, we have $C \cup C_1 \cup C_2 = C$, and conclude by \ruleref{t-ba-fun}.
  In other cases, based on the definition of $\fEmbed{t}$ we have
  $\cv{t_f^\prime} = \cv{t_x} \cup \set{x_f} \cup \cv{t_z} / \set{x, x^\prime, z}$.
  Therefore, $[x_f \mapsto C]\cv{t_f^\prime} = \cv{t_x}/\set{x^\prime} \cup C \cup \cv{t_z}/{z} / {x} = C \cup C_1 \cup C_2[x \mapsto C] / \set{\thisx}$.
  Also, $\cv{t_f^\prime} = C_1 \cup C_2 \cup \set{\thisx}$.
  We again conclude by the \ruleref{t-ba-fun} rule.

  \textsc{Case \ruleref{ba-tfun}}.
  As above.

  \textsc{Case \ruleref{ba-box}}.
  By IH we have $\adptcDft{\tCap{C}{S}}{\tCap{C^\prime}{S^\prime}}{\cCat}{C_0}$,
  $\cCat = \fCat{t_x}$,
  and $C_0[x] = \cv{t_x}$.
  Proceed by case analysis on $\cCat$.
  Each case can be concluded by definitions and the \ruleref{t-ba-box} rule.

  \textsc{Case \ruleref{ba-unbox}}.
  Conclude by IH and the \ruleref{t-ba-unbox} rule.
\end{proof}

\begin{theorem}[Completeness]  \label{thm:sub-adpt-completeness}
	If $\adpTypDft{x}{U}{t_x}$,
  then exists $C$, such that
  $\adptTypDft{x}{U}{\fCat{t_x}}{C[x]}$.
\end{theorem}

\begin{proof}
  Conclude by invoking Lemma \ref{lemma:sub-adapt-completeness-inner}.
\end{proof}

\subsubsection{Properties of Typing with Type-Level Box Adaptation}

\begin{lemma}  \label{lemma:typ-ubt-soundness}
  If $\typUbtDft{x}{T}{C}$, $\typUbDft{x}{t_x}{T}$ and $\cv{t_x} = C$.
\end{lemma}

\begin{proof}
  We prove by case analysis and in all cases we conclude immediately.
\end{proof}

\begin{lemma}  \label{lemma:adp-typ-cat}
  If $\adpTypDft{t}{t^\prime}{T}$, $\fCat{t} = \fCat{t^\prime}$.
\end{lemma}

\begin{proof}
  We prove by case analysis on the premise.
  For all cases, we can conclude immediately.
\end{proof}

\begin{theorem}[Soundness]  \label{thm:typ-adpt-soundness}
  If $\typAdptDft{t}{T}{C}$ then $\typAlgoDft{t}{t^\prime}{T}$ and $C = \cv{t^\prime}$.
\end{theorem}

\begin{proof}
  Proof proceeds by induction on the derivation.

  \emph{Case \ruleref{t-bi-var}, \ruleref{t-bi-box} and \ruleref{t-bi-unbox}}.
  We conclude immediately by the corresponding rule.

  \emph{Case \ruleref{t-bi-abs}}.
  We conclude by IH and \ruleref{bi-abs} rule.

  \emph{Case \ruleref{t-bi-tabs}}.
  We conclude by IH and \ruleref{bi-tabs} rule.

  \emph{Case \ruleref{t-bi-app}}.
  We first invoke Lemma \ref{lemma:typ-ubt-soundness} and show that
  $\typUbDft{x}{t_x}{C_x\,\tForall{z}{U}{T}}$ and $\cv{t_x} = C_1$.
  We then invoke Lemma \ref{thm:sub-adpt-soundness} and show that
  $\adpTypDft{y}{t_y}{U}$ and $\cv{t_y} = C_2$.
  Proceed by case analysis on the value of $\cCat$.
  If $\cCat = \cVar$, we can conclude immediately by applying \textsc{App}.
  Otherwise, we can show that $\avoidOp{z}{\cv{U}}{T} = \avoidOp{y'}{\cv{U}}{\substIn{z}{y^\prime}{T}}$
  and conclude by \ruleref{bi-app} rule.

  \emph{Case \ruleref{t-bi-tapp}}.
  We can prove this case similarly to the \ruleref{bi-app} rule.

  \emph{Case \ruleref{t-bi-let}}.
  By IH we can show that
  $\typAlgoDft{s}{s^\prime}{U}$
  and $\typAlgo{\G, x: U}{u}{u^\prime}{T}$.
  Proceed by case analysis on the category of $s$.
  If $s$ is a value and $x \notin \cv{t}$, by Lemma \ref{lemma:adp-typ-cat} we can show that
  $s^\prime$ is also a value.
  We can show that $\cv{t} = \cv{u^\prime}$.
  We can then conclude by \textsc{Let} rule.
  In other cases, we also invoke Lemma \ref{lemma:adp-typ-cat} and show that $\cv{t} = \cv{u^\prime} \cup \cv{s^\prime} / \set{x}$ and conclude by \ruleref{bi-let} rule.

\end{proof}

\begin{lemma} \label{lemma:typ-ubt-completeness}
  If $\typUbDft{x}{t_x}{T}$, $\typUbtDft{x}{T}{C}$ and $\cv{t_x} = C$.
\end{lemma}

\begin{proof}
  Proceed by case analysis on the premise.
  All cases can be concluded immediately.
\end{proof}

\begin{theorem}[Completeness]  \label{thm:typ-adpt-completeness}
  If $\typAlgoDft{t}{t^\prime}{T}$ then $\typAdptDft{t}{T}{C}$ and $C = \cv{t^\prime}$.
\end{theorem}

\begin{proof}
  Proof proceeds by induction on the derivation of the premise.

  \emph{Case \ruleref{bi-var}, \ruleref{bi-box} and \ruleref{bi-unbox}}.
  Conclude by the corresponding rule.

  \emph{Case \ruleref{bi-abs}, \ruleref{bi-tabs}}.
  We prove by IH and the corresponding rule.

  \emph{Case \ruleref{bi-app}}.
  We first invoke Lemma \ref{lemma:typ-ubt-completeness}
  and show that $\typUbtDft{x}{T}{C_1}$ and $C_1 = \cv{t_x}$.
  Then, we invoke Lemma \ref{thm:sub-adpt-completeness}
  and show that $\adptTypDft{y}{U}{\cCat}{C_2}$ and $C_2 = \cv{t_y}$.
  We can first show that $\cv{t} = \cv{t_x} \cup \cv{t_y} = C_1 \cup C_2$.
  If $t_y$ is a variable, we conclude immediately by the \ruleref{t-bi-app} rule.
  Otherwise, we can show that $\avoidOp{y^\prime}{\cv{U}}{\substIn{z}{y^\prime}{T}} = \avoidOp{z}{\cv{U}}{T}$.
  We then conclude by \ruleref{t-bi-app} rule.

  \emph{Case \ruleref{bi-tapp}}.
  We prove this case similarly to the \ruleref{bi-app} case.

  \emph{Case \ruleref{bi-let}}.
  By IH we have $\adptTypDft{s}{U}{C_1}$ and $C_1 = \cv{s^\prime}$;
  $\adptTypDft{u}{T}{C_2}$ and $C_2 = \cv{u^\prime}$.
  We can apply \ruleref{t-bi-let} and prove that $\typAdptDft{\tLet{x}{s}{u}}{T^\prime}{C^\prime}$.
  If $s$ is a value and $x \notin \cv{t^\prime}$, we have $\cv{t} = \cv{u^\prime} = C^\prime$ and conclude.
  Otherwise we have $\cv{t} = \cv{u^\prime} \cup \cv{s^\prime} / \set{x} = C^\prime$ and conclude.

\end{proof}

\subsection{Termination of Box Inference}
\label{sec:termination-proof}

\begin{wide-rules}

\textbf{Algorithmic typing \quad $\typAlgDft{t}{T}$}

\begin{multicols}{2}

\infrule[\ruledef{fs-var}]
  {\G(x) = S}
  {\typAlgDft{x}{S}}

\infrule[\ruledef{fs-abs}]
{\typAlg{\extendG{x}{U}}{t}{T}}
{\typAlgDft{\tLambda{x}{U}{t}}{{\tForall{x}{U}{T}}}}

\infrule[\ruledef{fs-tabs}]
{\typAlg{\extendGT{X}{S}}{t}{T}}
{\typAlgDft{\tTLambda{X}{S}{t}}{{\tTForall{X}{S}{T}}}}

\infrule[\ruledef{fs-app}]
{\typUpDft{x}{{\tForall{z}{T}{U}}} \\
  \typAlgDft{y}{T^\prime}\\
  {\subAlgoDft{T^\prime}{T}}}
{\typAlgDft{x\ y}{{U}}}

\infrule[\ruledef{fs-tapp}]
{\typUpDft{x}{{\tTForall{X}{S}{T}}} \\
  {\subAlgoDft{S^\prime}{S}}}
{\typAlgDft{x\ S^\prime}{\tSubst{X}{S^\prime}{U}}}

\end{multicols}

\infrule[\ruledef{fs-let}]
{\typAlgDft{s}{S} \\ \typAlg{\extendG{x}{S}}{t}{R}}
{\typAlgDft{\tLet{x}{s}{t}}{R}}

\textbf{Algorithmic subtyping \quad $\subAlgoDft{T}{U}$}

\begin{multicols}{2}

\infax[\ruledef{fs-refl}]
{\subAlgoDft{{X}}{{X}}}

\infrule[\ruledef{fs-tvar}]
{X <: S \in \G \\ {\subAlgoDft{S}{S^\prime}}}
{\subAlgoDft{X}{S^\prime}}

\infax[\ruledef{fs-top}]
{\subAlgoDft{S}{\top}}

\infrule[\ruledef{fs-fun}]
{\subAlgoDft{U_2}{U_1} \\
  \subAlgo{\extendG{x}{U_2}}{T_1}{T_2}}
{\subAlgoDft{\tForall{x}{U_1}{T_1}}{\tForall{x}{U_2}{T_2}}}

\infrule[\ruledef{fs-tfun}]
{\subAlgoDft{S_2}{S_1} \\
  \subAlgo{\extendGT{X}{S_2}}{T_1}{T_2}}
{\subAlgoDft{\tTForall{X}{S_1}{T_1}}{\tTForall{X}{S_2}{T_2}}}

\end{multicols}

\caption{Algorithmic Type for System F$_{<:}$}
  \label{fig:fsub-typing}

\end{wide-rules}

We first define the procedures for typechecking terms in both System $F_{<:}$ and System \ccAdp{}.
We begin by presenting the typing rules of System F$_{<:}$.
Figure \ref{fig:fsub-typing} shows the algorithmic typing rules of System F$_{<:}$.
The typechecking procedures $\typecheck{\G}{\cdot}$ and $\subtype{\G}{\cdot}{\cdot}$
can be straightforwardly defined based on the rules.
It picks the applicable rule by pattern-matching on the terms and the types.
The the term or the subtype relation is rejected if none of the rules matches.
In a similar vein,
we define $\typecheckAdp{\G}{\cdot}$ and $\subtypeAdp{\G}{\cdot}{\cdot}$ for System \ccAdp{}.
The subcapture procedure $\subcapture{\G}{C_1}{C_2}$ can also be straightforwardly defined
based on the subcapturing rules.
The basic idea is that we first inspect the size of $C_1$,
checking the subcapture $\subcapture{\G}{x}{C_2}$ for each element $x$ in $C_1$ if $C_1$ has more than one element,
and return true in the trivial case when $C_1$ is empty.
When $C_1 = \set{x}$ for some $x$,
we first attempt the \ruleref{sc-elem} rule which requires $x$ to be an element of $C_2$
and fallback to try the \ruleref{sc-var} rule by looking up the capture set $C_x$ of $x$ in the environment
and recursively invoking $\subcapture{\G}{C_x}{C_2}$.

Then, we define the simple-formed judgment.
Types with multiple levels of boxes
(i.e. $C_1\, \Hsquare \, C_2\, \Hsquare\, \cdots\, \Hsquare\, C_n \top$)
are impractical in the practice,
as a boxed type is meant to be pure
and it is illogical to attach a capture set to a boxed type and box it again.
However, this is allowed in the formulation of \cc{}.
Therefore, in the development of our theory,
we would like to ignore these nested boxed types
for the ease of reasoning and presentation.
To achieve this,
we introduce the concept of \emph{simple-formed} types, which are defined as follows,
and we assume that all types are simple-formed.
\begin{definition}[Simple-Formed Types]
	A type $T$ is simple-formed, if it is in one of the form:
  \begin{itemize}
  \item
    A shape type $S$, and its components are simple-formed too;

  \item
    $\tCap{C}{S}$, and $S$ is not a boxed type and it is simple-formed.
  \end{itemize}
\end{definition}


\begin{lemma}  \label{lemma:subcapture-termination}
	Given a well-formed context $\G$,
  and capture sets $C_1, C_2$ well-formed in $\G$,
  $\subcapture{\G}{C_1}{C_2}$ terminates.
\end{lemma}

To show termination, we first assign \emph{depths} to bindings.
The function $\bdepth{\G}{x}$ is defined as the following:
\begin{itemize}
\item
  If $\G = \Delta, x : T$,
  $\bdepth{\G}{x} = | \Delta |$.

\item
  Otherwise, if $\G = \Delta, y : T$,
  $\bdepth{\G}{x} = \bdepth{\Delta}{x}$.
\end{itemize}
Then, we prove the following auxillary, which states that
in a well-formed environment $\G$,
the depth of the captured references of a binding always decreases.

\begin{lemma}  \label{lemma:subcapt-depth-decr}
  For any well-formed environment $\G$,
  if $\G(x) = \tCap{C}{S}$,
  for any references $y \in C$
  we have $\bdepth{\G}{y} < \bdepth{\G}{x}$.
\end{lemma}

\begin{proof}
  By induction on the size of $\G$.
  If $\G$ is empty, it is trivial.
  Otherwise, $\G = \Delta, x : \tCap{C}{S}$.
  For any $y$ in the environment $\G$,
  if $x \neq y$,
  we have $\bdepth{\G}{y} = \bdepth{\Delta}{y}$.
  Similarly, based on the well-formedness of the environment,
  we know that $x \notin \cv{\G(y)}$.
  We can therefore conclude by IH.

  Otherwise, if $x = y$, we know that $\forall z \in C$, $z \in \dom{\Delta}$.
  Since $\bdepth{\G}{z} \leq | \Delta |$, while $\bdepth{\G}{x} = | \Delta | + 1$,
  we conclude immediately.
\end{proof}

Now we are ready for the proof of Lemma \ref{lemma:subcapture-termination}.
\begin{proof}
  We show that the algorithm $\subcapture{\G}{C_1}{C_2}$
  terminates in the lexical order of
  $(\max_{x \in C_1}\bdepth{\G}{x}, | C_1 |)$.
  We proceed by case analysis.

  If $C_1 = \set{x}$, and $x \in C_2$, the algorithm terminates immediately.
  Otherwise, if $C_1 = \set{x}$ but $x \notin C_2$,
  we invoke $\subcapture{\G}{C_1^\prime}{C_2}$
  where $C_1^\prime = \cv{\G(x)}$.
  By Lemma \ref{lemma:subcapt-depth-decr}
  we have $\max{x \in C_1^\prime}{\bdepth{\G}{x}} < \bdepth{\G}{x}$,
  which decreases in lexical order defined above.
  Finally, if $C_1 = \set{x_1, \cdots, x_n}$
  we will invoke
  $\subcapture{\G}{x_1}{C_2} \cdots \subcapture{\G}{x_n}{C_2}$.
  In each case we have $\max{z \in x_i}{\bdepth{\G}{x_i}} \leq \max{z \in C_1}{\bdepth{\G}{z}}$,
  and also $|\set{x_i}| < |C_1|$ which decreases in the lexical order too.
\end{proof}

\subadpterm*

\begin{proof}
	Proof by induction on the level of recursion
  when executing $\subtype{\eraseCC{\G}}{\eraseCC{T}}{\eraseCC{U}}$.
  In each case we analyze the boxed status of $T$ and $U$.

  If both $T$ and $U$ are capturing types, i.e.
  $T = \tCap{C_1}{S_1}$ and $U = \tCap{C_2}{S_2}$,
  we know that $\eraseCC{T} = \eraseCC{S_1}$
  and $\eraseCC{U} = \eraseCC{S_2}$.
  Also, neither $S_1$ nor $S_2$ are boxed types.
  In all cases we conclude by IH and Lemma \ref{lemma:subcapture-termination}.

  In another case, $T$ is boxed, while $U$ is not.
  $T = \tBox{\tCap{C_1}{S_1}}$ and $U = \tCap{C_2}{S_2}$.
  In this case, we apply the \adprn{unbox} rule first.
  This either fails, in the case that $C_1 \subseteq \dom{\G}$ does not hold,
  which terminates immediately.
  Otherwise it recurses, which can be concluded by IH.

  If $U$ is boxed, while $T$ is not,
  we prove this case similarly to the previous one,
  but the \adprn{box} rule is invoked.

  In the last case, where both $T$ and $U$ are boxed.
  $T = \tBox{\tCap{C_1}{S_1}}$ and $U = \tBox{\tCap{C_2}{S_2}}$.
  In this case, \adprn{boxed} is invoked first,
  after which we conclude with IH.
\end{proof}

\typadpterm*

\begin{proof}
	We proceed by induction on the level of recursion in $\typecheck{\eraseCC{\G}}{\eraseCC{t}}$.
  Each cases can be concluded by IH.
  In the \algorn{let} case,
  it is straightforward to see that $\avoidOp{x}{C}{T}$ always terminates,
  since what it does is traversing the types.
  In the \algorn{app} case we apply Lemma \ref{lemma:sub-adp-termination} together with the IHs to conclude.
\end{proof}

\end{document}